\documentclass[runningheads]{llncs}
\usepackage[T1]{fontenc}
\usepackage{graphicx}
\usepackage{cite}
\usepackage{amsmath,amssymb,amsfonts}
\usepackage{booktabs}
\usepackage{tabularx}
\usepackage[colorlinks=true,
            linkcolor={red!70!black},
            citecolor={green!70!black},
            urlcolor={blue!80!black}]{hyperref}
\usepackage[all]{hypcap}
\usepackage{tikz}
\usetikzlibrary{arrows.meta,positioning,fit,shapes.misc,calc,backgrounds}
\spnewtheorem{assumption}{Assumption}{\bfseries}{\itshape}

\usepackage{pgfplots}
\pgfplotsset{compat=1.18}
\usepackage{xcolor}
\usepackage{orcidlink}

\definecolor{sciBlue}{HTML}{0077BB}
\definecolor{sciOrange}{HTML}{EE7733}
\definecolor{sciTeal}{HTML}{009988}
\definecolor{sciRed}{HTML}{CC3311}
\definecolor{sciCyan}{HTML}{33BBEE}
\definecolor{sciMagenta}{HTML}{EE3377}
\definecolor{sciPurple}{HTML}{AA4499}
\definecolor{sciGray}{HTML}{555555}

\pgfplotsset{
    every axis/.append style={
        tick label style={font=\small},
        label style={font=\small},
        legend style={font=\small, fill=white, fill opacity=0.9, draw opacity=1, text opacity=1},
        grid style={densely dotted, color=gray!90},
        thick,
        enlargelimits=0.05,
        xtick={1, 3, 5, 7, 9, 11, 13, 15, 17, 19},
    }
}
\title{Obscura: Privacy-Preserving Protocol for the Algorand Blockchain Using LSAG Ring Signatures}
\titlerunning{Obscura: Privacy-Preserving Protocol for Algorand}
\author{Navid Azimi\inst{1}\orcidlink{0009-0009-1980-8521} \and Seyedamin Pouriyeh\inst{2}\orcidlink{0000-0002-5746-2914}}
\authorrunning{N. Azimi and S. Pouriyeh}
\institute{Emory University, Atlanta, GA 30322, USA  \\
\email{navid.azimi@emory.edu} \and
Kennesaw State University, Kennesaw, GA 30144, USA \\
\email{spouriye@kennesaw.edu}
}
\begin{document}
\maketitle
\begin{abstract}
While public blockchains provide transparent and auditable transaction histories, they inherently compromise user privacy. Existing privacy-enhancing protocols, such as those deployed on Ethereum, typically rely on succinct zero-knowledge proofs (zk-SNARKs) to obscure the transaction graph. However, implementing comparable cryptographic guarantees on high-throughput blockchains like Algorand is challenging due to strict per-call execution budgets and the state contention introduced by global Merkle accumulators. This paper presents \emph{Obscura}, a decentralized, non-custodial privacy protocol tailored for constrained smart contract environments. Obscura achieves transaction anonymity using Linkable Spontaneous Anonymous Group (LSAG) signatures over the BN254 elliptic curve, verified entirely on-chain. To overcome limitations of the Algorand Virtual Machine (AVM), we introduce a novel state model that leverages Algorand's Box Storage for $O(1)$ commitment membership checks, eliminating the need for global Merkle accumulators, and a dynamic opcode-budget expansion mechanism via pooled inner application calls. Our implementation demonstrates that signer-ambiguous privacy is practical and efficient on Algorand without relying on trusted setups or succinct proofs. Obscura provides a robust privacy layer for transparent ledgers, bridging the gap between high-throughput blockchain architectures and the dual requirements of cryptographic privacy and selective auditability.
\keywords{Algorand \and Linkable Ring Signatures \and LSAG \and Blockchain Privacy \and Anonymous Transactions \and On-chain Verification \and Smart Contracts}
\end{abstract}
%
\section{Introduction}
The transparent nature of public blockchains \cite{nakamoto2008bitcoin} is fundamental to their security and audibility, yet it inherently undermines user privacy. This friction is especially pronounced on high-throughput blockchains such as Algorand \cite{gilad2017algorand}, which are increasingly used for enterprise financial services and decentralized finance (DeFi). In these settings, participants require confidentiality to protect sensitive transaction data, while still maintaining the ability to meet regulatory and audit requirements. Because all transactions are recorded on a public ledger, adversaries can trace the flow of funds, cluster addresses, and deanonymize users through transaction graph analysis \cite{ron2013quantitative,meiklejohn2013fistful,androulaki2013evaluating}. To address this, privacy-preserving protocols and decentralized mixers have been developed \cite{miers2013zerocoin,sasson2014zerocash,bonneau2014mixcoin}. On Ethereum \cite{wood2014ethereum}, protocols such as Tornado Cash \cite{pertsev2019tornado} successfully deployed privacy pools that sever the on-chain link between depositors and withdrawers. These systems typically rely on zk-SNARKs to prove membership in a Merkle tree \cite{merkle1987digital} of deposits without revealing which specific fund is being spent. Related smart-contract designs include M\"obius \cite{meiklejohn2018mobius}, which employs linkable ring signatures for trustless tumbling on Ethereum, and Zether \cite{bunz2020zether}, which provides confidential account balances.

However, the deployment of SNARK-based privacy pools relies heavily on complex cryptographic circuits and state-heavy global accumulators. Emerging high-throughput blockchains, such as Algorand, optimize for fast finality and predictable fees by imposing strict per-transaction execution budgets (opcode limits) and flat state models. While the Algorand Virtual Machine (AVM) supports native pairing operations, evaluating the multiple pairings required for traditional zk-SNARK verification (e.g., Groth16 \cite{groth2016size}) consumes a massive portion of the pooled execution budget. Furthermore, SNARK-based mixers require maintaining a global Merkle tree, which introduces severe state-contention bottlenecks on high-throughput ledgers and relies on trusted setups. This architectural divergence creates a significant gap in the ecosystem: achieving non-custodial transaction anonymity on highly constrained smart contract platforms without relying on trusted setups or global state accumulators.

\paragraph{Proposed Approach.}
To bridge this gap, we present \emph{Obscura}, a privacy-preserving protocol that achieves transaction anonymity on Algorand without relying on succinct proofs or trusted setups. The core idea of Obscura is that Linkable Spontaneous Anonymous Group (LSAG) signatures \cite{liu2004linkable}, combined with scalable key-value storage and dynamic execution budget allocation, can provide a practical privacy layer on smart contract platforms where pairing-based succinct verification is costly and global accumulators are ill-suited to the state model.

In Obscura, a user deposits a fixed denomination by publishing a cryptographic commitment, a public key derived from a secret scalar on the BN254 elliptic curve. To withdraw the funds to a new address, the user generates an LSAG signature over a ring of on-chain commitments (including their own and several decoys). This signature proves knowledge of the secret corresponding to one of the commitments without revealing which one, thereby providing \emph{signer ambiguity}. To prevent double-spending, the signature includes a \emph{key image} \cite{van2013cryptonote} (or nullifier), which is deterministically derived from the user's secret. The smart contract verifies the LSAG signature natively using AVM's elliptic-curve opcodes and ensures the key image has not been previously recorded.

\paragraph{Key Challenges and Technical Solutions.}
Implementing an LSAG-based privacy-preserving protocol on Algorand presents two primary technical challenges: state management and execution constraints.

First, traditional mixers use a global Merkle tree to track deposits, which requires hashing operations for every insertion and Merkle path verifications during withdrawal. Obscura eschews the Merkle accumulator entirely. Instead, it leverages Algorand's \emph{Box Storage} \cite{algorand_box_storage}, a scalable key-value store, to record individual commitments and nullifiers. Membership of each ring point is verified in $O(1)$ time by asserting the existence of its corresponding box, significantly simplifying the state model and reducing cryptographic overhead.

Second, verifying an LSAG signature requires $O(n)$ elliptic curve scalar multiplications and additions, where $n$ is the ring size. This linear complexity quickly exceeds the AVM's strict per-call opcode budget. Obscura overcomes this limitation through a dynamic budget expansion strategy. By issuing a burst of inner application calls to a dedicated, stateless ``dummy'' application before the main verification loop, the protocol aggregates the opcode budgets of multiple grouped transactions into a single execution context. This allows the complete verification of the ring signature within a single block.

\paragraph{Novelty and Positioning.}
To the best of our knowledge, no prior work has investigated a non-custodial privacy pool for Algorand in which linkable ring signatures are verified entirely on-chain by an AVM smart contract. We emphasize that this claim does not concern the underlying cryptographic primitives: LSAG signatures \cite{liu2004linkable} and key-image-based double-spend prevention \cite{van2013cryptonote,noether2016ring} are well established, Box Storage is a native Algorand facility \cite{algorand_box_storage}, and opcode budget pooling is a documented feature of the AVM \cite{algorand_tx_fees}. Rather, the novelty lies in how these components are combined and adapted to the execution and state model of the AVM. Concretely, Obscura differs from existing work in the following respects:
\begin{itemize}
    \item \emph{Proof system and setup.} SNARK-based designs such as Zerocash \cite{sasson2014zerocash} and Tornado Cash \cite{pertsev2019tornado} prove membership in a Merkle tree of deposits with a succinct proof whose verification relies on pairings and a trusted setup. Obscura instead uses an LSAG signature over an explicit, user-chosen ring, requiring neither a trusted setup nor a succinct proof, at the cost of $O(n)$ proof size and a bounded ring size.
    \item \emph{State model.} Rather than maintaining a global accumulator whose root must be updated on every deposit, Obscura stores each commitment and each key image in its own box and validates every ring member individually by box existence. No Merkle path is computed or verified, and deposits do not contend for a single shared root.
    \item \emph{Execution environment.} CryptoNote-style systems \cite{van2013cryptonote,noether2016ring} verify linkable ring signatures at the consensus layer of a dedicated blockchain, and M\"obius \cite{meiklejohn2018mobius} verifies them within Ethereum smart contracts under gas metering. Obscura targets the AVM, whose binding constraints are a fixed per-call opcode budget, a 2048-byte application-argument limit, and per-transaction box-reference limits. Accordingly, we derive the exact per-member verification cost (7,608 opcodes), the minimum inner-call multiplier that covers it ($M=12$), and the maximum ring size admitted by the argument limit ($n=19$).
    \item \emph{Application setting.} Obscura targets fixed-denomination, signer-ambiguous transfers on an account-based, high-throughput ledger, with user-controlled selective disclosure through the per-deposit secret $x$. This differs from the goal of Zether \cite{bunz2020zether}, which focuses on confidential account balances.
\end{itemize}
We do not claim that Obscura provides stronger anonymity than accumulator-based mixers; as discussed in Sections~\ref{sec:sec} and~\ref{sec:eval}, its anonymity set is bounded by the ring size, whereas that of a Merkle-based pool is bounded by the number of deposits.

\paragraph{Contributions.}
This paper makes the following contributions:
\begin{itemize}
    \item \textbf{Protocol Design:} We formalize Obscura, an LSAG-based privacy protocol tailored for the constraints of the Algorand Virtual Machine, demonstrating that transaction anonymity is achievable without SNARKs or trusted setups.
    \item \textbf{State and Execution Model:} We introduce an $O(1)$ per-member ring validation approach using Box Storage in place of a global Merkle accumulator, and we adapt the AVM's opcode-pooling mechanism to LSAG verification by deriving the exact per-member cost and the minimum deterministic inner-call multiplier, which together circumvent the AVM's state and execution bottlenecks.
    \item \textbf{End-to-End Implementation:} We provide a fully functional, open-source implementation of the protocol, including the PyTeal smart contract, an off-chain Python prover, and a client-side application. We also detail practical privacy-hardening techniques, such as recent-biased decoy selection and cryptographically secure shuffling, to mitigate operational deanonymization risks.
\end{itemize}

\medskip
\noindent
Figure~\ref{fig:arch} illustrates the end-to-end protocol architecture. 

The remainder of this paper is organized as follows: Section~\ref{sec:prelim} introduces the cryptographic and system preliminaries that establish the foundation for the subsequent construction. Building on these preliminaries, Section~\ref{sec:crypto} formalizes the cryptographic construction, while Section~\ref{sec:protocol} describes the resulting protocol and its operation. Section~\ref{sec:onchain} then presents the smart contract execution model and opcode pooling strategy, followed by Section~\ref{sec:impl}, which details the system implementation. With the construction and implementation established, Section~\ref{sec:sec} analyzes its security and privacy properties under an explicit threat model. Section~\ref{sec:eval} subsequently defines the evaluation criteria, evaluates the protocol's performance and practical efficiency. Finally, Section~\ref{sec:concl} concludes the paper and discusses directions for future research.

\begin{figure}[t]
  \centering
  \includegraphics[width=\textwidth]{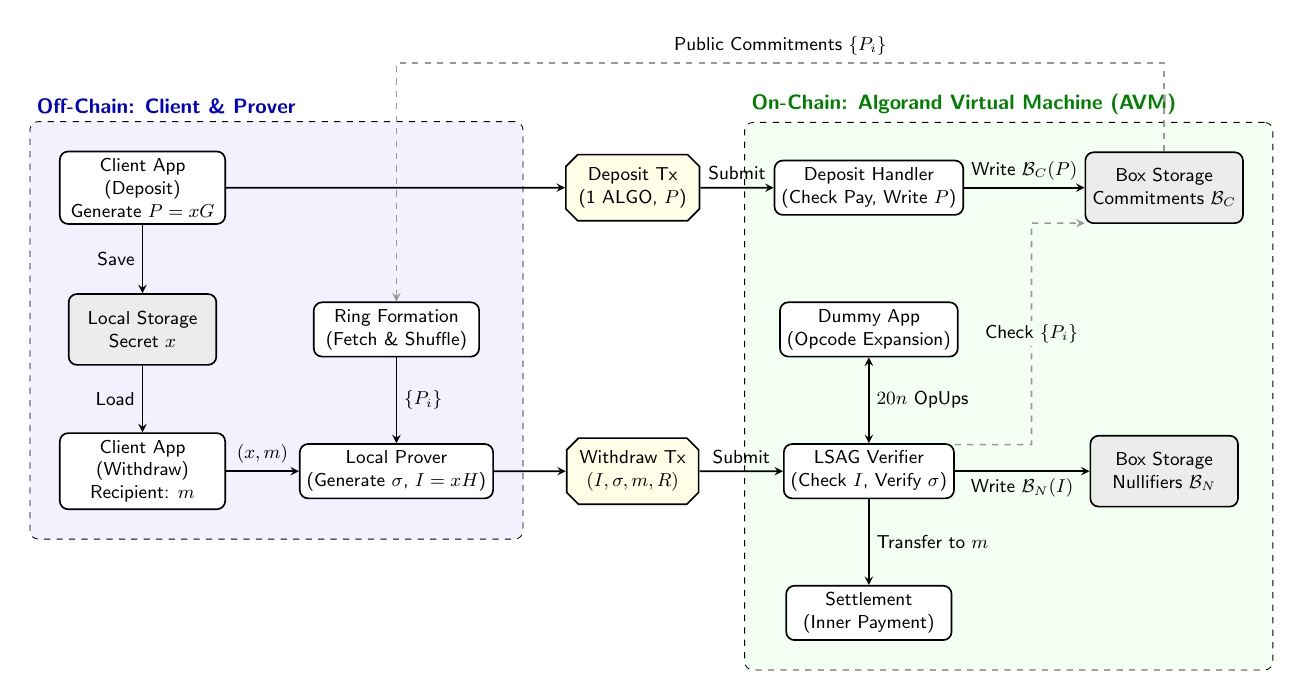}
  \caption{Obscura protocol end-to-end architecture. Solid arrows denote protocol actions and transaction submissions; dashed arrows denote reads of on-chain public state. The Client App initiates both phases. \textbf{Deposit:} the client samples a secret scalar $x$, computes the public commitment $P = xG$, stores $x$ in local storage, and submits a deposit transaction $(1~\mathrm{ALGO}, P)$. The on-chain Deposit Handler verifies the payment and writes $P$ into commitment Box Storage $\mathcal{B}_C$. \textbf{Withdrawal:} the client loads $x$ and specified recipient $m$, and fetches public commitments $\{P_i\}$ from $\mathcal{B}_C$ to form a shuffled ring $R$. THe local prover then produces an LSAG signature $\sigma$ and key image (nullifier) $I = xH$. The client submits a withdrawal transaction $(I, \sigma, m, R)$. The LSAG Verifier reads $\{P_i\}$ from $\mathcal{B}_C$, expands its opcode budget via inner calls to a Dummy App ($20n$ OpUps), verifies $\sigma$ natively with AVM elliptic-curve operations, records $I$ in nullifier Box Storage $\mathcal{B}_N$ if unspent, and settles by an inner payment to $m$.}
  \label{fig:arch}
\end{figure}

\section{Preliminaries}
\label{sec:prelim}
This section formalizes the cryptographic primitives and the blockchain execution model that underpin the Obscura protocol.

\subsection{Cryptographic Foundations}

\paragraph{Elliptic Curve Assumptions.}
Let $\mathbb{G}$ be a prime-order subgroup of the BN254 elliptic curve \cite{bos2014elliptic,barreto2005pairing}, with order $q$. We rely on the hardness of the Elliptic Curve Discrete Logarithm Problem (ECDLP) in $\mathbb{G}$ \cite{koblitz1987elliptic,miller1985use}: given a generator $G \in \mathbb{G}$ and a point $P = xG$, it is computationally infeasible to recover the scalar $x \in \mathbb{Z}_q$. Obscura utilizes two independent, publicly verifiable generators $G, H \in \mathbb{G}$ for which the discrete logarithm relationship is unknown.

\paragraph{Linkable Ring Signatures.}
A standard ring signature \cite{rivest2001leak} allows a user to sign a message on behalf of an ad hoc set of public keys (the \emph{ring} or \emph{anonymity set}) without revealing which specific key was used. Unlike group signatures \cite{chaum1991group}, ring signatures are \emph{spontaneous}: no group manager, setup procedure, or coordination among ring members is required. Discrete-logarithm-based instantiations \cite{abe2002one} realize the ring structure through proofs of partial knowledge \cite{cramer1994proofs}. While this provides signer ambiguity, it lacks a mechanism to prevent a user from signing multiple times using the same secret key.

\emph{Linkable Spontaneous Anonymous Group} (LSAG) signatures \cite{liu2004linkable} extend this concept by introducing a \emph{key image} (or \emph{linking tag}). For a secret key $x$ and its corresponding public key (commitment) $P = xG$, the key image is deterministically derived as $I = xH$. This key-image mechanism was popularized by the CryptoNote \cite{van2013cryptonote} and underpins the MLSAG \cite{noether2016ring} and CLSAG \cite{goodell2019concise} constructions deployed in Monero. The LSAG construction ensures three critical security properties:
\begin{enumerate}
    \item \textbf{Anonymity (Signer Ambiguity):} Given a signature, a ring of size $n$, and a key image $I$, an adversary cannot determine the signer's index with probability significantly greater than $1/n$ \cite{liu2005linkable}.
    \item \textbf{Unforgeability:} It is computationally infeasible to produce a valid signature for a ring without knowing the discrete logarithm of at least one public key in that ring \cite{liu2004linkable}.
    \item \textbf{Linkability:} Any two signatures generated using the same secret key $x$ will necessarily produce the identical key image $I$. In a decentralized ledger, this property is leveraged to prevent double-spending by ensuring $I$ is recorded and checked against future transactions \cite{noether2016ring}.
\end{enumerate}
The non-interactive nature of the signature is achieved via the Fiat--Shamir heuristic \cite{fiat1986prove}, modeling the cryptographic hash function as a random oracle \cite{bellare1993random}.

\subsection{Blockchain Execution Model}

While public blockchains provide robust consensus and transparency, they expose the entire transaction graph to public scrutiny. Privacy-preserving smart contracts must obscure this graph while operating within the strict resource bounds of the consensus layer. Obscura is designed for the Algorand Virtual Machine (AVM) \cite{algorand_avm}, which imposes specific architectural considerations:

\paragraph{Execution Constraints and Opcode Budget.}
Smart contracts on the AVM execute within a strict, per-transaction computational limit (opcode budget). While the AVM natively supports elliptic-curve operations on BN254 (e.g., \texttt{EcAdd}, \texttt{EcScalarMul}), verifying an LSAG signature requires $O(n)$ such operations. To circumvent the per-transaction budget limit, the AVM allows \emph{inner transactions}, application calls made by the contract itself. By issuing a burst of inner calls to a stateless ``dummy'' application, a contract can dynamically pool the opcode budgets of multiple grouped transactions into a single execution context, enabling the verification of large rings.

\paragraph{State Management via Box Storage.}
Traditional privacy pools maintain a global Merkle tree of deposits to verify membership succinctly. However, maintaining a global state accumulator introduces storage bottlenecks and requires users to provide Merkle inclusion proofs. Obscura utilizes Algorand's \emph{Box Storage}, an unbounded key-value store scoped to the application. Commitments and key images are stored directly as discrete boxes. This allows the smart contract to verify the membership of any public key in $O(1)$ time by simply asserting the existence of its corresponding box, fundamentally altering the state architecture from a global accumulator to a flat, verifiable key-value mapping.

\section{Cryptographic Construction}
\label{sec:crypto}
This section formalizes the Obscura protocol's instantiation of Linkable Spontaneous Anonymous Group (LSAG) signatures over the BN254 elliptic curve. We define the construction in terms of three primary algorithms: \textsf{KeyGen}, \textsf{Sign}, and \textsf{Verify}.

\subsection{System Parameters and Hash Instantiation}
Let $\mathbb{G}$ be the prime-order subgroup of the BN254 elliptic curve, with prime order $q \approx 2^{254}$. We define two independent, publicly verifiable generators $G, H \in \mathbb{G}$. The discrete logarithm $\log_G H$ is unknown. Points in $\mathbb{G}$ are serialized as 64-byte uncompressed big-endian coordinates to align with the Algorand Virtual Machine (AVM) \texttt{EcAdd} and \texttt{EcScalarMul} opcodes \cite{avm_opcodes,algorand_avm}.

We require a cryptographic hash function modeled as a random oracle \cite{bellare1993random}, $\mathcal{H}: \{0,1\}^* \to \mathbb{Z}_q$. To map the output into a scalar within the AVM's computational constraints, the protocol instantiates $\mathcal{H}$ using $\mathrm{SHA\text{-}256}$ \cite{nist2015sha} followed by a bitwise AND with a fixed 255-bit mask. This clears the most significant bit of the 256-bit digest, constraining the challenge to a 255-bit integer; scalars are subsequently reduced modulo $q$ by the curve arithmetic. Appendix~\ref{app:formal} analyzes the induced challenge distribution and shows that this instantiation preserves the security of the Fiat--Shamir transform.

\subsection{Key Generation (\textsf{KeyGen})}
To participate in the protocol (i.e., make a deposit), a user executes the randomized key generation algorithm:
\begin{enumerate}
    \item Sample a secret scalar $x \xleftarrow{\$} \mathbb{Z}_q$.
    \item Compute the public commitment $P = xG \in \mathbb{G}$.
    \item Compute the key image (nullifier) $I = xH \in \mathbb{G}$.
\end{enumerate}
The user publishes $P$ to the smart contract during the deposit phase. The secret key $x$ and the key image $I$ are kept private until the withdrawal phase.

\subsection{Signature Generation (\textsf{Sign})}
To withdraw funds anonymously, the user acts as a prover to generate an LSAG signature $\sigma$. Let $m \in \{0,1\}^{256}$ be the transaction message, defined in Obscura as the 32-byte public key of the withdrawal recipient. This binds the signature to the specific recipient, preventing replay attacks.

The prover selects an anonymity set (ring) of $n$ public commitments $R = \{P_0, \dots, P_{n-1}\}$ currently stored on-chain, ensuring their own commitment $P$ is included at a secret index $\pi \in [0, n-1]$ such that $P_\pi = xG$. The signature generation proceeds as follows:

\begin{enumerate}
    \item \textbf{Initialization:} Sample a random scalar $\alpha \xleftarrow{\$} \mathbb{Z}_q$. Compute the initial curve points for the true signer index:
    \begin{align*}
        L_\pi &= \alpha G \\
        R_\pi &= \alpha H
    \end{align*}
    Compute the subsequent challenge $c_{\pi+1} = \mathcal{H}(m \| L_\pi \| R_\pi)$.
    
    \item \textbf{Decoy Simulation:} For each decoy index $i = \pi+1, \dots, n-1, 0, \dots, \pi-1$ (evaluating indices modulo $n$):
    \begin{enumerate}
        \item Sample a random response $s_i \xleftarrow{\$} \mathbb{Z}_q$.
        \item Compute the curve points using the running challenge $c_i$:
        \begin{align*}
            L_i &= s_i G + c_i P_i \\
            R_i &= s_i H + c_i I
        \end{align*}
        \item Compute the next challenge $c_{i+1} = \mathcal{H}(m \| L_i \| R_i)$.
    \end{enumerate}
    
    \item \textbf{Closure:} Close the ring by computing the response for the true signer index $\pi$:
    \[
        s_\pi = \alpha - c_\pi x \pmod q
    \]
\end{enumerate}
The resulting signature is the tuple $\sigma = (c_0, s_0, \dots, s_{n-1})$.

\subsection{Signature Verification (\textsf{Verify})}
The smart contract acts as the verifier. It receives the ring $R$, the key image $I$, the message $m$, and the signature $\sigma$. Verification requires $O(n)$ elliptic curve operations and proceeds linearly:

\begin{enumerate}
    \item Initialize the running challenge $c = c_0$.
    \item For $i = 0, \dots, n-1$:
    \begin{enumerate}
        \item Compute $L_i = s_i G + c P_i$.
        \item Compute $R_i = s_i H + c I$.
        \item Update the challenge $c \leftarrow \mathcal{H}(m \| L_i \| R_i)$.
    \end{enumerate}
    \item \textbf{Accept} if and only if the final challenge $c$ equals the initial $c_0$.
\end{enumerate}

The algebraic correctness relies on the fact that at the true index $\pi$, the verifier computes $L_\pi = s_\pi G + c_\pi P_\pi = (\alpha - c_\pi x)G + c_\pi(xG) = \alpha G$, and similarly $R_\pi = \alpha H$, matching the prover's initialization.

\subsection{Proof Serialization}
To optimize on-chain parsing, the signature and ring are serialized into a tightly packed byte array (\texttt{zk\_proof}) structured as:
\begin{equation}
    \underbrace{n}_{\text{1 byte}} \ \|\ \underbrace{P_0 \| \dots \| P_{n-1}}_{n \times 64 \text{ bytes}} \ \|\ \underbrace{c_0}_{\text{32 bytes}} \ \|\ \underbrace{s_0 \| \dots \| s_{n-1}}_{n \times 32 \text{ bytes}}
    \label{eq:pack}
\end{equation}
The total payload size is exactly $96n + 33$ bytes. The smart contract extracts $c_0$ and the responses $\{s_i\}$ directly from this payload, executing the \textsf{Verify} algorithm iteratively while dynamically expanding its opcode budget via inner transactions.

\section{Protocol Description}
\label{sec:protocol}
The Obscura protocol operates in two primary phases: deposit (commitment) and withdrawal (anonymous spending). We formalize the transaction semantics, detailing the strict separation between off-chain cryptographic generation and on-chain state transitions.

\subsection{Deposit Phase}
The deposit phase allows a user to insert funds into the anonymity pool by registering a public commitment. To ensure atomicity between the fund transfer and the state update, the deposit is executed as an \emph{atomic transaction group} on the Algorand blockchain.

\paragraph{Off-Chain Execution.}
The user samples a secret scalar $x \xleftarrow{\$} \mathbb{Z}_q$ and computes the public commitment $P = xG$. The user securely stores $x$ locally; it is never transmitted to the network.

\paragraph{On-Chain State Transition.}
The user submits an atomic group consisting of two transactions:
\begin{enumerate}
    \item A smart contract application call invoking the \textsf{Deposit} method with argument $P$.
    \item A payment transaction transferring exactly $10^6$ micro-ALGO (1 ALGO) to the smart contract's escrow address.
\end{enumerate}
Upon receiving the deposit transaction group, the smart contract verifies the payment amount and the receiver. It then checks the global Box Storage for the existence of a commitment box $\mathcal{B}_C$ keyed by a prefix of the commitment, specifically \texttt{"c"}$\,\|\,P[0\!:\!32]$. If the box already exists, the transaction aborts. Otherwise, the contract allocates the box, writes the full 64-byte $P$ as its value, and increments a global deposit counter. The complete sequence of operations for the deposit phase is illustrated in Figure~\ref{fig:seq_deposit}.

\begin{figure}[htbp]
\centering
\includegraphics[width=0.88\textwidth]{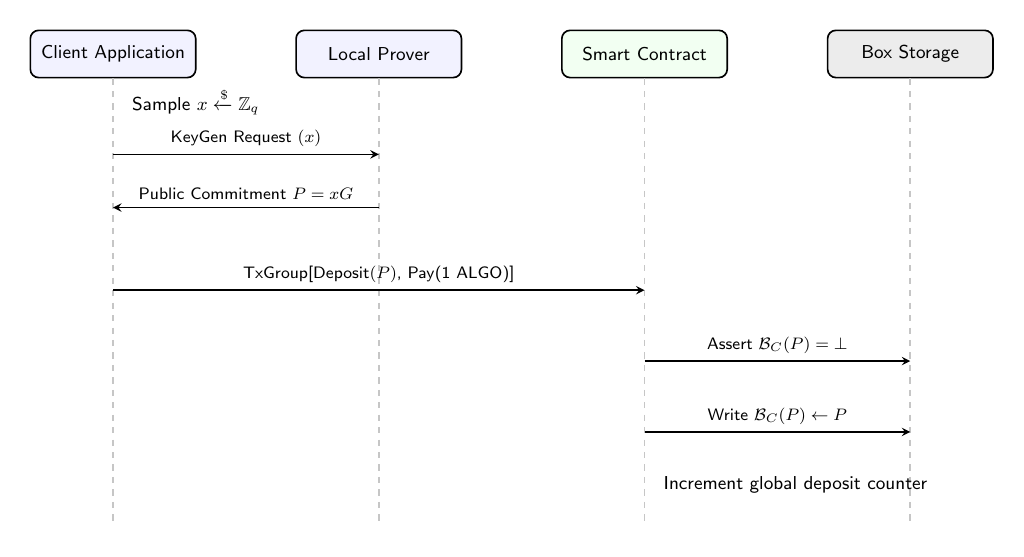}
\caption{Deposit phase sequence diagram. The client application generates a secret scalar and delegates the \textsf{KeyGen} elliptic curve operations to the local prover. The client then submits an atomic transaction group to the smart contract, which verifies the payment and allocates a new commitment box in on-chain storage.}
\label{fig:seq_deposit}
\end{figure}

\subsection{Anonymity Set Construction}
Prior to withdrawal, the user must construct an anonymity set (ring) $R = \{P_0, \dots, P_{n-1}\}$. This process occurs entirely off-chain to prevent the network from learning the true signer's index.

The client queries the blockchain's indexer to retrieve a set of public commitments. To mitigate temporal intersection attacks, the client selects decoys using a recency-biased distribution, favoring recently added commitments; empirical studies of deployed ring-signature systems show that decoy sampling that fails to mimic realistic spending behavior enables statistical deanonymization \cite{moser2018empirical,kumar2017traceability}. The client then injects their own commitment $P$ into the set and applies a Fisher--Yates shuffle \cite{fisher1953statistical}. This ensures the true commitment is uniformly distributed within the ring, providing optimal $1/n$ signer ambiguity.

\subsection{Withdrawal Phase}
The withdrawal phase allows a user to spend a previously deposited fund to a new recipient address $m$ without revealing which commitment is being spent.

\paragraph{Off-Chain Execution.}
The user computes the key image $I = xH$ and executes the \textsf{Sign} algorithm defined in Section~\ref{sec:crypto} over the ring $R$ and message $m$. The resulting signature $\sigma$ is packed into a compact byte array.

\paragraph{On-Chain State Transition.}
The user submits a transaction group (consisting of a main application call and optionally auxiliary calls to supply additional box references if $n$ is large) invoking the \textsf{Withdraw} method, providing the key image $I$, the packed signature $\sigma$, the recipient address $m$, and an array of box references corresponding to $I$ and all $P_i \in R$. 

Upon receiving the transaction, the smart contract validates the anonymity set against on-chain storage, verifies the LSAG signature natively, and ensures the key image $I$ has not been previously recorded. The precise on-chain execution pipeline and opcode pooling strategy required to perform this validation are detailed in Section~\ref{sec:onchain}. Upon successful verification, the contract allocates the nullifier box $\mathcal{B}_N$ to permanently record $I$ as spent, and issues an inner payment transaction transferring the deposit amount (minus protocol and storage fees) to the recipient $m$.

Figure~\ref{fig:seq_withdraw} details the complete sequence of off-chain and on-chain interactions during the withdrawal phase.

\begin{figure}[htbp]
\centering
\includegraphics[width=\textwidth]{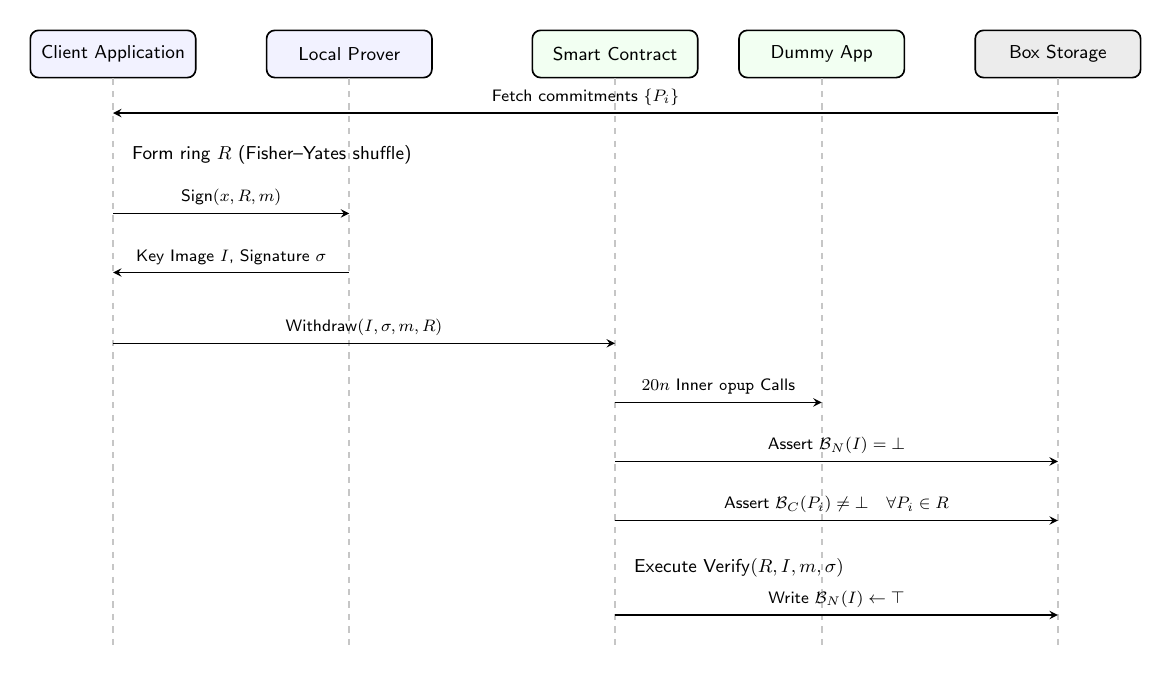}
\caption{Withdrawal phase sequence diagram. The client application constructs an anonymity set from on-chain public commitments and delegates signature generation to the local prover. The smart contract dynamically expands its opcode budget via inner calls to a dummy application, validates the ring members and key image against Box Storage, verifies the LSAG signature, and settles the payment.}
\label{fig:seq_withdraw}
\end{figure}

\subsection{Double-Spend Prevention and Linkability}
The protocol's security relies on the strict decoupling of the deposit and withdrawal phases. The only cryptographic link between a deposit $P$ and a withdrawal is the shared secret $x$, knowledge of which is proven in zero-knowledge \cite{goldwasser1989knowledge} via the LSAG signature. 

Double-spending is deterministically prevented by the algebraic properties of the key image $I = xH$. Because $H$ is a fixed generator, any attempt to spend $P$ a second time will inevitably yield the identical key image $I$. The smart contract's enforcement of the $\mathcal{B}_N$ box uniqueness guarantees that $I$ acts as a globally unique, unforgeable serial number for the underlying deposit, preserving ledger integrity without compromising signer ambiguity.

\section{Smart Contract Execution}
\label{sec:onchain}
Executing complex cryptographic protocols on high-throughput blockchains requires reconciling the computational intensity of the verifier with the strict execution limits of the underlying virtual machine. This section details how Obscura implements the \textsf{Verify} algorithm within the Algorand Virtual Machine (AVM), leveraging specific opcode strategies to ensure privacy-preserving transaction validation.

\subsection{Dynamic Opcode Pooling and Deterministic Budgeting}
To maintain high network throughput, the AVM imposes a strict computational limit, referred to as the \emph{opcode budget}, on every smart contract invocation. A standard application call is allocated a base budget of 700 units \cite{algorand_sc_costs_constraints}. However, verifying a single member of an LSAG ring over the BN254 curve requires four scalar multiplications (\texttt{EcScalarMul}) and two point additions (\texttt{EcAdd}). In the AVM, opcode costs are strictly constant and deterministic; they do not vary based on network congestion, input values, or blockchain state. Table~\ref{tab:opcode_costs} provides the exact mathematical breakdown of the AVM opcode costs required to verify a single ring member, totaling 7,608 opcodes.

To overcome this limitation without requiring interactive verification or off-chain SNARK proofs, Obscura employs a dynamic budget expansion strategy known as \emph{opcode pooling}. The AVM allows the execution budget of multiple transactions within an atomic group (including inner transactions) to be pooled into a shared execution context \cite{algorand_tx_fees}. 

Before executing the cryptographic verification loop, the Obscura smart contract issues $12n$ inner application calls to a stateless, foreign ``dummy'' application. This dummy application exposes a single \texttt{opup} method that performs no logic and simply approves the call. Because each inner transaction contributes an additional 700 units to the shared pool, $12n$ calls dynamically provision sufficient computational headroom to verify a ring of size $n$ entirely on-chain. This trades a marginally higher transaction fee for the ability to execute non-interactive, single-epoch verification.
\begin{table}[ht]
\centering
\caption{Deterministic AVM Opcode Costs for LSAG Verification (Per Ring Member)}
\label{tab:opcode_costs}
\begin{tabularx}{\textwidth}{@{}X r@{}}
\toprule
\textbf{Operation} & \textbf{Opcode Cost} \\
\midrule
$4 \times$ \texttt{EcScalarMul} (BN254g1) ($4 \times 1810$) & 7,240 \\
$2 \times$ \texttt{EcAdd} (BN254g1) ($2 \times 125$) & 250 \\
$1 \times$ \texttt{Sha256} ($1 \times 35$) & 35 \\
$1 \times$ \texttt{BytesAnd} ($1 \times 6$) & 6 \\
$1 \times$ \texttt{App.box\_length} ($1 \times 1$) & 1 \\
Loop Overhead and Control Flow & $\sim$76 \\
\midrule
\textbf{Total Exact Cost per Member ($C$)} & \textbf{7,608} \\
\bottomrule
\end{tabularx}
\vspace{0.5cm}

\caption{Inner Transaction Budget Gain and Multiplier Calculation}
\label{tab:budget_gain}
\begin{tabularx}{\textwidth}{@{}X r@{}}
\toprule
\textbf{Budget Parameter} & \textbf{Value (Opcodes)} \\
\midrule
Budget Added per Inner Call & +700 \\
Main Loop Inner Txn Submit Overhead & -23 \\
\midrule
\textbf{Net Budget Gain per Call ($G$)} & \textbf{677} \\
\bottomrule
\end{tabularx}
\end{table}

To increase the available budget, the contract executes a loop that constructs and submits inner application calls. As detailed in Table~\ref{tab:budget_gain}, building and submitting an inner transaction introduces an overhead of 23 opcodes in the main transaction. Since each inner transaction adds exactly 700 opcodes of pooled budget, the net budget gain per inner application call is 677 opcodes.

Let $n$ be the ring size and $M$ be the inner transaction multiplier (the number of inner calls per ring member). The total opcode cost is $7608n + 23Mn + 115$ (where 115 accounts for non-loop initialization overhead), and the total opcode budget provided is $700 + 700Mn$. For successful execution, the budget must strictly exceed the cost:
\begin{align*}
700 + 700Mn &> 7608n + 23Mn + 115 \\
677Mn &> 7608n - 585 \\
677M &> 7608 - \frac{585}{n}
\end{align*}

For the maximum practical ring size $n = 19$, this simplifies to $677M > 7577.2$, requiring $M > 11.19$. Because the multiplier must be an integer, the absolute theoretical minimum multiplier is $M = 12$. At this setting, each member is allocated a budget of 8,400 opcodes (net gain of 8,124 opcodes), providing a precise 516-opcode (7.2\%) safety margin over the required 7,608 opcodes. By utilizing this exact deterministic minimum, the protocol guarantees successful execution while achieving the absolute lowest possible user transaction fees.

\subsection{On-Chain Verification Pipeline}
The on-chain verification pipeline strictly orders state validation and cryptographic execution to ensure soundness and prevent resource exhaustion attacks. The execution path for the withdrawal proceeds as follows:

\begin{enumerate}
    \item \textbf{Budget Provisioning:} The contract executes the $12n$ inner \texttt{opup} calls to expand the available opcode budget.
    \item \textbf{Double-Spend Assertion:} The contract evaluates the key image $I$ provided in the transaction arguments. It executes a state read against the Box Storage to assert that the nullifier box $\mathcal{B}_N$ (\texttt{"n"}$\,\|\,I[0\!:\!32]$) has a length of zero (i.e., it does not exist). If the box exists, the execution halts immediately, preventing double-spending.
    \item \textbf{Anonymity Set Validation:} The contract parses the packed proof payload to extract the ring size $n$ and the array of public commitments $\{P_0, \dots, P_{n-1}\}$. For each $P_i$, the contract asserts the existence of its corresponding commitment box $\mathcal{B}_C$ (\texttt{"c"}$\,\|\,P_i[0\!:\!32]$). This enforces that the anonymity set is composed exclusively of valid deposits, preventing an adversary from injecting fabricated public commitments to trivially forge a signature.
    \item \textbf{Cryptographic Loop:} The contract initializes the running challenge $c = c_0$. For each index $i \in [0, n-1]$, it executes the core LSAG verification logic using native AVM opcodes:
    \begin{align*}
        L_i &= \texttt{EcAdd}(\texttt{EcScalarMul}(G, s_i), \texttt{EcScalarMul}(P_i, c)) \\
        R_i &= \texttt{EcAdd}(\texttt{EcScalarMul}(H, s_i), \texttt{EcScalarMul}(I, c)) \\
        h_i &= \mathrm{SHA\text{-}256}(m \| L_i \| R_i) \\
        c' &= \texttt{BytesAnd}(h_i, \text{mask})
    \end{align*}
    The running challenge is updated ($c \leftarrow c'$). The \texttt{BytesAnd} opcode applies a fixed 255-bit mask to the hash digest, ensuring the resulting scalar is strictly within the field $\mathbb{Z}_q$.
    \item \textbf{Closure and Settlement:} After the loop terminates, the contract asserts that the final challenge $c$ strictly equals the initial $c_0$ provided in $\sigma$. Upon successful closure, the contract allocates $\mathcal{B}_N$ to record the key image $I$ and issues the inner payment transaction to the recipient $m$.
\end{enumerate}

\subsection{Execution-Layer Atomicity}
The AVM's execution model guarantees that state transitions are strictly atomic. If any assertion fails, whether due to an invalid signature closure, a non-existent commitment box, or an already-spent key image, the entire transaction group reverts. This ensures that the nullifier is never recorded unless the signature is cryptographically valid, and funds are never disbursed unless the nullifier is successfully recorded. Furthermore, because the recipient address $m$ is explicitly bound into the hash chain $h_i$, any attempt by an adversary or a malicious node to intercept the transaction and alter the payout destination will deterministically break the signature closure, causing the execution to revert.

\section{System Implementation}
\label{sec:impl}
The Obscura protocol is realized as a full-stack decentralized application, bridging off-chain cryptographic generation with on-chain verification. The implementation is partitioned into three distinct layers: the on-chain smart contract, the off-chain cryptographic prover, and the client-side application. This architecture strictly isolates sensitive key material from public network transmission while optimizing for the execution constraints of the Algorand Virtual Machine (AVM). Figure~\ref{fig:impl_arch} illustrates the system's architecture and trust boundaries.

\begin{figure}[htbp]
\centering
\includegraphics[width=0.95\textwidth]{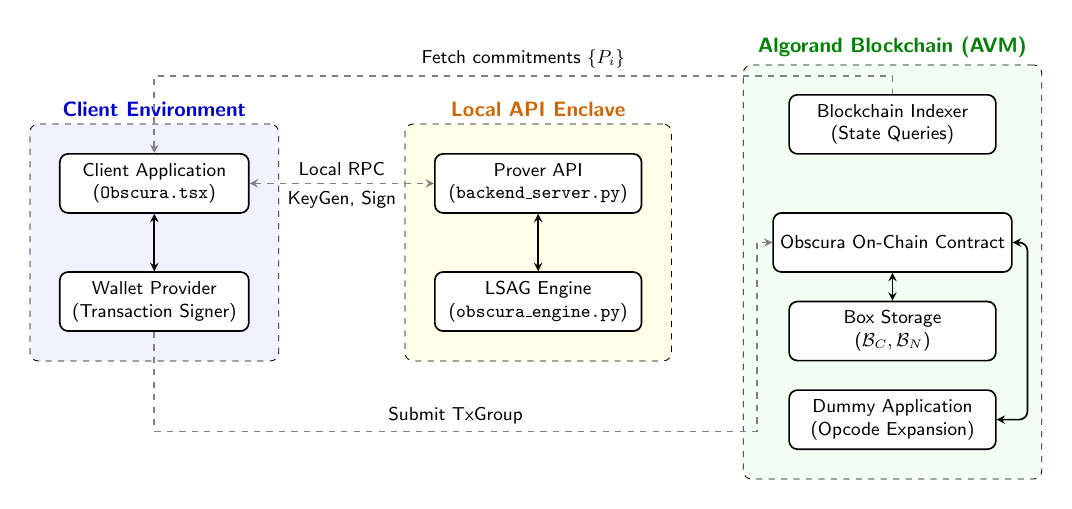}
\caption{Obscura implementation architecture. The system is partitioned into three distinct trust domains. The client environment handles decoy selection and transaction construction. The local API enclave executes heavy elliptic curve operations without exposing the secret scalar $x$ to the network. The Algorand blockchain enforces state transitions, utilizing the indexer for state queries, Box Storage for membership checks, and a dummy application for opcode pooling.}
\label{fig:impl_arch}
\end{figure}

\subsection{Smart Contract Layer}
The verification logic and state management are implemented in PyTeal \cite{algorand_pyteal} and compiled to AVM bytecode. The contract acts as the ultimate arbiter of the protocol, enforcing the atomic state transitions defined in Section~\ref{sec:protocol}. 

To circumvent the storage limitations of global and local state schemas, the contract utilizes Algorand Box Storage. Commitment boxes and nullifier boxes are instantiated using 33-byte keys, prefixed with \texttt{"c"} and \texttt{"n"} respectively, concatenated with the first 32 bytes of the curve point, mapping to their full 64-byte values. This flat, key-value architecture enables $O(1)$ membership and double-spend checks, entirely avoiding the computational overhead of maintaining and verifying Merkle tree accumulators on-chain.

\subsection{Off-Chain Prover Enclave}
Heavy elliptic curve operations required for signature generation are delegated to a dedicated Python-based prover engine. This engine implements the \textsf{KeyGen} and \textsf{Sign} algorithms and exposes them via a stateless HTTP prover API. 

When a user initiates a withdrawal, the client application transmits the secret scalar $x$, the selected anonymity set $R$, and the recipient address $m$ to the prover. The prover computes the key image $I$ and the LSAG signature $\sigma$, returning them to the client. In a production environment, this prover must be operated as a trusted local enclave (e.g., compiled to WebAssembly within the browser) to ensure the secret scalar $x$ is never exposed to an untrusted network.

\subsection{Client Application and Decoy Selection}
The frontend is a React/TypeScript application that manages user interaction, wallet integration (via the Pera Wallet connector \cite{pera_connect}), and transaction construction. 

A critical security responsibility of the client application is the construction of the anonymity set. To mitigate temporal intersection attacks, the client queries the blockchain indexer for recent deposits, maintaining a recency-biased pool of commitments. It selects $n-1$ decoys and injects the user's true commitment $P$. To prevent positional bias that could deanonymize the user, the client applies a Fisher--Yates shuffle to the ring. Crucially, this shuffle is seeded using the cryptographically secure \texttt{crypto.getRandomValues()} API, thereby avoiding the use of weak pseudorandom number generators \cite{mdn_crypto_getrandomvalues}.

\subsection{Engineering Constraints and Parameterization}
To minimize transaction size and parsing overhead on-chain, the signature $\sigma$ and the ring $R$ are serialized into a tightly packed byte array before transmission. As defined in Equation~\ref{eq:pack}, the payload strictly follows the format: a 1-byte ring size $n$, followed by $n$ 64-byte commitments, a 32-byte initial challenge $c_0$, and $n$ 32-byte scalar responses. This results in a deterministic payload size of $96n + 33$ bytes.

The ring size $n$ is a configurable protocol parameter, taking values from $n=1$ to $n=19$. This upper limit is dictated by the AVM consensus parameter \texttt{MaxAppTotalArgLen}, which restricts the total size of application arguments to 2048 bytes. Because the proof payload requires exactly $96n + 33$ bytes, $n=19$ provides the maximum anonymity set that fits within this limit ($1857$ bytes for the proof, plus additional routing and state arguments). While smaller than the anonymity sets found in SNARK-based protocols, this parameterized design allows users to select an optimal balance between signer ambiguity and the linear execution costs of native elliptic-curve operations on Algorand.

\section{Security and Privacy Analysis}
\label{sec:sec}
We evaluate the security and privacy guarantees of the Obscura protocol under a formal adversarial model, separating cryptographic properties from operational and system-level limitations. Formal definitions, theorems, and complete proofs are deferred to Appendix~\ref{app:formal}.

\subsection{Threat Model and Assumptions}
We assume a public, immutable ledger where all state transitions, smart contract payloads, and Box Storage contents are perfectly visible to a Globally Passive Adversary (GPA), following standard anonymity terminology \cite{pfitzmann2010terminology}. The adversary may also act actively by injecting malicious transactions, attempting to double-spend, or providing malformed cryptographic payloads. We assume the underlying Algorand consensus mechanism \cite{gilad2017algorand} and the AVM execution environment are correct and secure.

Crucially, we assume the off-chain \textsf{Sign} algorithm is executed within a trusted local enclave controlled by the user. Delegating proof generation to an untrusted remote server trivially violates privacy, as the server must observe the secret scalar $x$ to compute the signature.

\subsection{Cryptographic Guarantees}

\paragraph{Unforgeability and Soundness.}
The protocol's soundness relies on the existential unforgeability of the LSAG construction under chosen-message attacks (EUF-CMA) \cite{goldwasser1988digital}. Assuming the hardness of the Elliptic Curve Discrete Logarithm Problem (ECDLP) on the BN254 curve, and modeling the hash function $\mathcal{H}$ as a random oracle, it is computationally infeasible for an adversary to produce a valid signature $\sigma$ without knowledge of the secret scalar $x$ corresponding to at least one public commitment $P_i$ in the ring $R$ \cite{liu2004linkable}.

\paragraph{Non-Malleability and Replay Resistance.}
To prevent transaction malleability and replay attacks, the recipient's public key $m$ is explicitly bound into the Fiat--Shamir challenge chain: $h_i = \mathrm{SHA\text{-}256}(m \| L_i \| R_i)$. If an active network adversary intercepts a valid withdrawal transaction and attempts to redirect the funds by altering $m$, the hash outputs will diverge entirely. This breaks the algebraic closure of the ring (i.e., the final challenge will not equal $c_0$), causing the smart contract to deterministically reject the transaction.

\paragraph{Linkability and Double-Spend Resistance.}
Double-spending is prevented via the cryptographic property of linkability. The protocol guarantees that any two signatures generated from the same secret scalar $x$ will yield the identical key image $I = xH$. Because the AVM smart contract atomically asserts the non-existence of the nullifier box $\mathcal{B}_N$ prior to executing the payout, the system achieves deterministic double-spend resistance. An adversary cannot reuse a spent commitment, even if they select a completely different set of decoys for the second ring.

\subsection{Privacy Analysis and Practical Limitations}

\paragraph{Signer Ambiguity.}
The primary privacy guarantee of the protocol is signer ambiguity. Given a valid signature $\sigma$ over a ring $R$ of size $n$, a computationally bounded adversary cannot identify the true signer's index $\pi$ with a probability significantly greater than $1/n$. The true commitment $P$ and the key image $I$ are generated using independent base points ($G$ and $H$), rendering them computationally unlinkable without knowledge of $x$ under the Decisional Diffie--Hellman assumption in $\mathbb{G}$ (Theorem~\ref{thm:anonymity}, Appendix~\ref{app:formal}).

\paragraph{Auditability and Selective Disclosure.}
While Obscura provides default anonymity to the public, it inherently supports user-controlled selective transparency. Because the off-chain client retains the secret scalar $x$, a user can cryptographically prove the provenance and destination of their funds to an authorized third party (e.g., an auditor or regulatory body). By disclosing $x$, the auditor can deterministically verify both the deposit commitment $P = xG$ and the withdrawal key image $I = xH$. This mechanism ensures compatibility with compliance and reporting obligations without compromising the privacy of the broader anonymity set, in the spirit of accountable-privacy designs for decentralized anonymous payments \cite{garman2016accountable}.

\paragraph{Decoy Selection and Intersection Attacks.}
Theoretical anonymity assumes a uniform distribution of fetched commitments. In practice, adversaries can employ temporal intersection attacks or statistical correlation if users predictably select decoys or withdraw immediately after depositing; analogous heuristics have been used to trace a substantial fraction of transactions in early Monero \cite{moser2018empirical,kumar2017traceability} and to deanonymize careless users of Ethereum mixers \cite{beres2021blockchain}. The Obscura client mitigates this by enforcing a recency-biased decoy pool and applying a cryptographically secure shuffle. This ensures the true commitment's position in the on-chain array reveals no information. However, users must still exercise user-side operational discipline (e.g., introducing time delays between deposits and withdrawals) to maximize the effective anonymity set.

\paragraph{Anonymity Set Size Constraints.}
Due to the 2048-byte application argument limit of the AVM, the current testnet implementation caps the ring size at $n=19$. While this provides functional privacy against casual blockchain surveillance and is a significant improvement over smaller rings, the combinatorial hiding remains weaker than SNARK-based accumulators (e.g., Tornado Cash), where the effective anonymity set is bounded only by the number of active deposits in a pool. Obscura trades the theoretical unbounded size of the anonymity set for the ability to execute natively and transparently on a high-throughput, non-pairing-friendly virtual machine.

\paragraph{Post-Quantum Vulnerability.}
The protocol's security relies fundamentally on the hardness of ECDLP. A quantum adversary executing Shor's algorithm \cite{shor1994algorithms} could efficiently recover the secret scalar $x$ from the public commitment $P = xG$. With $x$ exposed, the adversary could compute $I = xH$ to deanonymize that specific transaction (if already spent) or forge a withdrawal for the corresponding commitment (if unspent). While the per-deposit sampling of $x$ isolates the compromise to individual deposits rather than a global identity, this remains a fundamental limitation of all classical elliptic-curve privacy protocols; transitioning to post-quantum security would require replacing the LSAG construction with lattice-based or hash-based alternatives \cite{esgin2019matrict,lu2019raptor}.

\section{Performance Evaluation}
\label{sec:eval}
We evaluate the practical resource footprint of the Obscura protocol, focusing on computational overhead, storage requirements, and the inherent scalability limitations of executing ring signatures natively on the AVM. We first describe the experimental setup and define the evaluation criteria, then position the evaluation with respect to prior work, and finally report and interpret the measurements.

\subsection{Experimental Setup}
All measurements were obtained with the reference implementation of Section~\ref{sec:impl} deployed on the Algorand Testnet: the PyTeal contract and the dummy \texttt{opup} application on-chain, the Python prover exposed through its local HTTP API, and the React client with Pera Wallet. We sweep the ring size over its entire admissible range, $n \in \{1, \dots, 19\}$, and execute complete deposit and withdrawal flows at each ring size. The client instruments each flow into the operational stages defined in Section~\ref{sec:eval_results}. Each flow is executed five times per ring size, and latencies are reported as the mean over these five runs, with error bars denoting one standard deviation.
Proof size, inner-transaction count, and transaction fees are deterministic functions of $n$ and are therefore reported without error bars.

\subsection{Evaluation Criteria}
\label{sec:eval_criteria}
Because LSAG proof size and verification cost both scale linearly with $n$, the main question for Obscura is not whether a single configuration is fast, but how each cost grows with the anonymity set and where it meets the hard limits of the Algorand blockchain. We therefore report every criterion as a function of $n$. Table~\ref{tab:eval_criteria} summarizes the criteria, which we define as follows.

\begin{itemize}
    \item \textbf{Anonymity set size ($n$).} The number of commitments in the ring, which bounds the probability of identifying the true signer by $1/n$ plus a negligible term (Theorem~\ref{thm:anonymity}). It serves as the independent variable of the evaluation, against which all costs are traded off. We do not empirically measure resistance to statistical deanonymization; as discussed in Section~\ref{sec:sec}, operational behavior can only degrade this cryptographic bound.
    \item \textbf{Proof size (bytes).} The length of the packed \texttt{zk\_proof} payload of Equation~\ref{eq:pack}. It measures the communication cost of a withdrawal and determines feasibility, since the payload must fit within the 2048-byte \texttt{MaxAppTotalArgLen} limit together with the remaining arguments. It is therefore the criterion that fixes the maximum ring size.
    \item \textbf{On-chain execution cost (opcodes and inner transactions).} The opcode cost of verification and the number of inner transactions the contract must issue to provision it. Because AVM opcode costs are deterministic, this criterion tests whether the analytical cost model of Section~\ref{sec:onchain} is sufficient for verification to complete in a single call at every admissible $n$. The inner-transaction count is also the quantity that drives the fee criterion below.
    \item \textbf{Transaction fees (ALGO).} The network fees incurred by the deposit and withdrawal transaction groups. This is the direct monetary cost borne by users, and it captures the economic price of opcode pooling, since each inner transaction incurs the minimum network fee.
    \item \textbf{Proof generation time (ms).} The time for the local prover to compute the key image and execute \textsf{Sign}. It measures the client-side computational burden, which is relevant because the security model requires signing to occur in a trusted local environment rather than on a remote server (Section~\ref{sec:sec}).
    \item \textbf{End-to-end latency and stage breakdown (ms).} The wall-clock time from the user initiating a flow to on-chain confirmation, decomposed into stages. It reflects the user-perceived responsiveness of the protocol, and the breakdown separates costs that Obscura controls (data preparation, proof generation, transaction construction) from network-dependent costs (submission and block confirmation).
    \item \textbf{On-chain state footprint and payout overhead.} The storage allocated per deposit and per withdrawal, the associated minimum balance requirement, and the resulting reduction of the withdrawal payout. These measure the long-term state cost of replacing an accumulator with per-item boxes and the effective value received by the user.
\end{itemize}

\begin{table}[ht]
\centering
\caption{Evaluation Criteria and Their Role in the Evaluation}
\label{tab:eval_criteria}
\begin{tabularx}{\textwidth}{@{}l X l@{}}
\toprule
\textbf{Criterion} & \textbf{What It Assesses} & \textbf{Reported In} \\
\midrule
Ring size $n$ & Signer ambiguity bound ($1/n$) & Independent variable \\
Proof size & Communication cost; feasibility under 2048-byte limit & Fig.~\ref{fig:proof_size} \\
Execution cost & Sufficiency of opcode pooling; inner calls required & Tables~\ref{tab:opcode_costs}, \ref{tab:budget_gain}; Fig.~\ref{fig:inner_txs} \\
Transaction fees & Monetary cost to the user & Fig.~\ref{fig:fees} \\
Proof generation & Client-side computational burden & Fig.~\ref{fig:proof_time} \\
Latency & User-perceived responsiveness, by stage & Figs.~\ref{fig:runtime_total}--\ref{fig:withdrawal_breakdown} \\
State and payout & Storage growth; value received by user & Section~\ref{sec:eval_overhead} \\
\bottomrule
\end{tabularx}
\end{table}

\subsection{Positioning of the Evaluation}
To the best of our knowledge, no prior work has reported an on-chain implementation and measurement of linkable ring signature verification on the AVM. Consequently, there is no like-for-like Algorand baseline against which Obscura can be compared. We also deliberately refrain from numerical comparisons with Ethereum-based designs such as Tornado Cash \cite{pertsev2019tornado} or M\"obius \cite{meiklejohn2018mobius}, as such comparisons would conflate differences in fee markets, execution metering (gas versus a per-call opcode budget), and block times, and we have not measured those systems under a common setup. Instead, we evaluate Obscura against two references that are well defined for our setting: the analytical cost model derived in Section~\ref{sec:onchain}, and the hard limits imposed by the AVM. This choice also reflects what distinguishes the design. The quantities measured here, namely the per-member opcode cost covered by inner \texttt{opup} calls, the box-based validation of every ring member, and the argument-size bound on $n$, arise from the design choices described in Section~\ref{sec:onchain}. They differ structurally from the costs of accumulator-based designs, in which membership is established by verifying a succinct proof against a Merkle root.

\subsection{Computational and Storage Overhead}
\label{sec:eval_overhead}

\paragraph{Proof Size and Communication Complexity.}
Unlike succinct zero-knowledge proofs (e.g., zk-SNARKs \cite{groth2016size}) which offer $O(1)$ proof sizes, LSAG signatures scale linearly with the anonymity set size $n$. As defined in Equation~\ref{eq:pack}, the serialized \texttt{zk\_proof} payload requires exactly $96n + 33$ bytes. For the maximum ring size of $n=19$, the payload is 1857 bytes, fitting just under the 2048-byte protocol limit. While this is efficient compared to non-elliptic-curve ring signatures, the $O(n)$ scaling imposes a strict upper bound on feasible ring size due to transaction size constraints and increasing network propagation overhead.

\paragraph{On-Chain State Footprint.}
Obscura eschews global state accumulators in favor of Algorand Box Storage. Each deposit allocates a 64-byte commitment box, and each withdrawal allocates a 64-byte nullifier box. This results in an $O(N)$ storage footprint, where $N$ is the total number of transactions in the system's history. While this prevents the computational bottleneck of Merkle tree updates, it shifts the burden to state bloat. To mitigate this, the protocol enforces a strict minimum balance requirement (MBR) for box allocation, which is deducted from the withdrawal payout.

\paragraph{Execution Cost and Opcode Scaling.}
The computational bottleneck of the protocol is the on-chain \textsf{Verify} algorithm. As mathematically proven in Section~\ref{sec:onchain}, each ring member strictly requires 7,608 opcode units for elliptic curve scalar multiplications and additions. To satisfy this, the client must provision $12n$ inner \texttt{opup} calls. This linear scaling of execution cost directly translates to higher transaction fees for the user, as each inner call incurs the network's minimum base fee. Figure~\ref{fig:inner_txs} illustrates the scaling of the required inner transactions, and Figure~\ref{fig:fees} demonstrates the resulting linear impact on the total withdrawal transaction fee, compared to the constant fee of the deposit phase.

\subsection{Scalability Trade-offs}
\label{sec:eval_results}

\begin{figure}[htbp]
\centering
\begin{tikzpicture}
\begin{axis}[
    width=10.64cm,
    height=6.08cm,
    xlabel={Ring Size},
    ylabel={Number of Inner Transactions},
    xmin=1, xmax=19,
    ymin=0, ymax=240,
    grid=major,
    legend pos=north west,
]

\addplot[
    color=sciTeal,
    thick,
    mark=*,
    mark options={solid, fill=sciTeal},
]
coordinates {
(1,13) (2,25) (3,37) (4,49) (5,61)
(6,73) (7,85) (8,97) (9,109) (10,121)
(11,133) (12,145) (13,157) (14,169) (15,181)
(16,193) (17,205) (18,217) (19,229)
};
\addlegendentry{Inner Transactions}

\end{axis}
\end{tikzpicture}
\caption{Number of inner transactions issued during the withdrawal phase versus ring size. To overcome the strict per-transaction AVM opcode budget, the smart contract dynamically groups inner \texttt{opup} calls. Each ring member strictly requires 12 inner transactions to safely cover the 7,608 opcode verification cost, leading to an $O(n)$ scaling in the total number of inner transactions.}
\label{fig:inner_txs}
\end{figure}
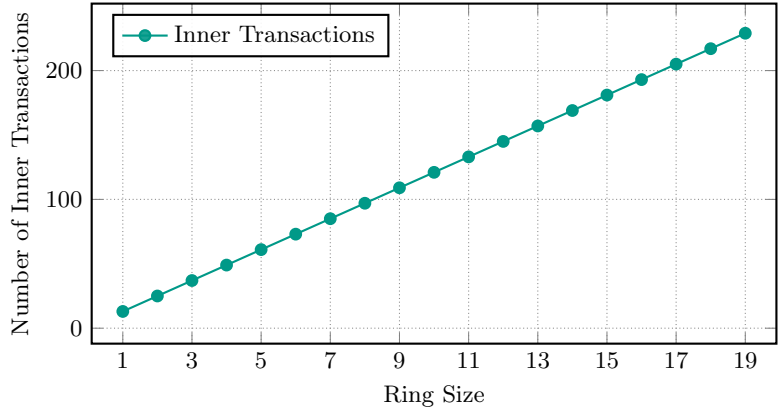

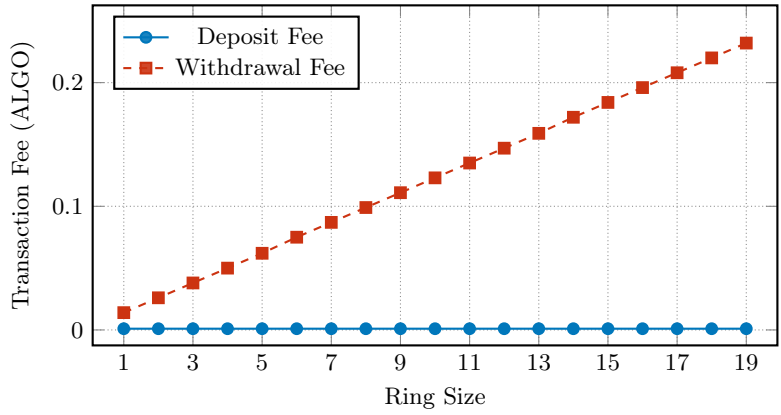
\begin{figure}[htbp]
\centering
\begin{tikzpicture}
\begin{axis}[
    width=10.64cm,
    height=6.08cm,
    xlabel={Ring Size},
    ylabel={Transaction Fee (ALGO)},
    xmin=1, xmax=19,
    ymin=0, ymax=0.25,
    grid=major,
    legend pos=north west,
]

\addplot[
    color=sciBlue,
    thick,
    mark=*,
    mark options={solid, fill=sciBlue},
]
coordinates {
(1,0.001) (2,0.001) (3,0.001) (4,0.001) (5,0.001)
(6,0.001) (7,0.001) (8,0.001) (9,0.001) (10,0.001)
(11,0.001) (12,0.001) (13,0.001) (14,0.001) (15,0.001)
(16,0.001) (17,0.001) (18,0.001) (19,0.001)
};
\addlegendentry{Deposit Fee}

\addplot[
    color=sciRed,
    thick,
    dashed,
    mark=square*,
    mark options={solid, fill=sciRed},
]
coordinates {
(1,0.014) (2,0.026) (3,0.038) (4,0.050) (5,0.062)
(6,0.075) (7,0.087) (8,0.099) (9,0.111) (10,0.123)
(11,0.135) (12,0.147) (13,0.159) (14,0.172) (15,0.184)
(16,0.196) (17,0.208) (18,0.220) (19,0.232)
};
\addlegendentry{Withdrawal Fee}

\end{axis}
\end{tikzpicture}
\caption{Deposit and withdrawal transaction fees as a function of ring size. The deposit fee remains constant at 0.001 ALGO because it requires a simple state box allocation. Conversely, the withdrawal fee scales linearly due to the dynamic opcode pooling mechanism, which issues $12n$ inner \texttt{opup} calls per ring member to provision sufficient AVM execution budget.}
\label{fig:fees}
\end{figure}

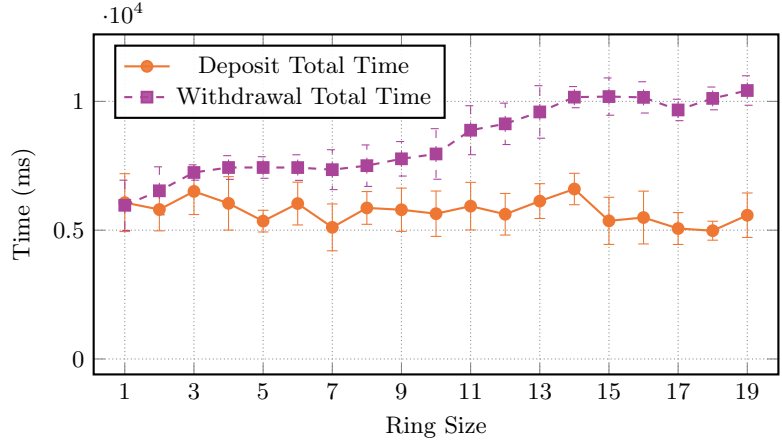
\begin{figure}[htbp]
\centering
\begin{tikzpicture}
\begin{axis}[
    width=10.64cm,
    height=6.08cm,
    xlabel={Ring Size},
    ylabel={Time (ms)},
    xmin=1, xmax=19,
    ymin=0, ymax=12000,
    grid=major,
    legend pos=north west,
]

\addplot[
    color=sciOrange, 
    thick, 
    mark=*, 
    mark options={solid, fill=sciOrange},
    error bars/.cd, y dir=both, y explicit
]
coordinates {
(1,6069.80) +- (0,1123.00) (2,5802.00) +- (0,827.92) (3,6505.00) +- (0,898.90) 
(4,6040.00) +- (0,1037.41) (5,5351.40) +- (0,421.00) (6,6033.20) +- (0,830.33) 
(7,5109.40) +- (0,909.47) (8,5863.20) +- (0,635.63) (9,5792.40) +- (0,839.60) 
(10,5639.40) +- (0,880.62) (11,5932.40) +- (0,924.09) (12,5618.60) +- (0,807.99) 
(13,6128.40) +- (0,675.44) (14,6598.00) +- (0,609.58) (15,5362.00) +- (0,915.08) 
(16,5490.60) +- (0,1024.81) (17,5064.20) +- (0,618.41) (18,4979.40) +- (0,367.90) 
(19,5580.40) +- (0,863.23)
};
\addlegendentry{Deposit Total Time}

\addplot[
    color=sciPurple, 
    thick, 
    dashed, 
    mark=square*, 
    mark options={solid, fill=sciPurple},
    error bars/.cd, y dir=both, y explicit
]
coordinates {
(1,5965.40) +- (0,975.37) (2,6528.40) +- (0,931.19) (3,7239.40) +- (0,297.55) 
(4,7432.20) +- (0,459.43) (5,7432.80) +- (0,418.36) (6,7430.20) +- (0,496.49) 
(7,7348.60) +- (0,774.43) (8,7504.00) +- (0,804.76) (9,7770.20) +- (0,669.37) 
(10,7958.80) +- (0,981.28) (11,8880.40) +- (0,950.76) (12,9127.80) +- (0,803.57) 
(13,9589.60) +- (0,1019.29) (14,10163.80) +- (0,410.43) (15,10186.20) +- (0,723.91) 
(16,10154.60) +- (0,606.76) (17,9666.60) +- (0,415.93) (18,10114.60) +- (0,442.12) 
(19,10421.40) +- (0,571.55)
};
\addlegendentry{Withdrawal Total Time}

\end{axis}
\end{tikzpicture}
\caption{End-to-end deposit and withdrawal latency across varying ring sizes. The deposit process exhibits a relatively stable latency profile regardless of the anonymity set size, as its operations are constant. In contrast, withdrawal latency increases linearly with the ring size, driven primarily by the $O(n)$ computational overhead of generating the zero-knowledge LSAG signature off-chain and managing larger transaction groups.}
\label{fig:runtime_total}
\end{figure}

\begin{figure}[htbp]
\centering
\begin{tikzpicture}
\begin{axis}[
    width=10.64cm,
    height=6.08cm,
    xlabel={Ring Size},
    ylabel={Proof Generation Time (ms)},
    xmin=1, xmax=19,
    ymin=0, ymax=4500,
    grid=major,
    legend pos=north west,
]

\addplot[
    color=sciTeal, 
    thick, 
    mark=*, 
    mark options={solid, fill=sciTeal},
    error bars/.cd, y dir=both, y explicit
]
coordinates {
(1,640.80) +- (0,150.09) (2,719.00) +- (0,118.64) (3,1013.00) +- (0,144.52) 
(4,1123.20) +- (0,146.84) (5,1351.40) +- (0,168.66) (6,1570.60) +- (0,129.62) 
(7,1622.60) +- (0,40.67) (8,2008.40) +- (0,135.61) (9,2114.00) +- (0,129.42) 
(10,2308.20) +- (0,129.86) (11,2417.20) +- (0,83.33) (12,2779.40) +- (0,138.08) 
(13,3000.20) +- (0,128.20) (14,3104.60) +- (0,189.70) (15,3313.60) +- (0,165.42) 
(16,3585.40) +- (0,111.39) (17,3715.80) +- (0,147.34) (18,3996.20) +- (0,160.06) 
(19,3951.60) +- (0,55.58)
};
\addlegendentry{Proof Generation Time}

\end{axis}
\end{tikzpicture}
\caption{Zero-knowledge ring signature proof generation time versus ring size. The off-chain local prover executes the \textsf{Sign} algorithm using BN254 elliptic curve operations. Since the prover must compute decoy responses for every member in the anonymity set to construct the LSAG signature, the generation time scales linearly $O(n)$.}
\label{fig:proof_time}
\end{figure}
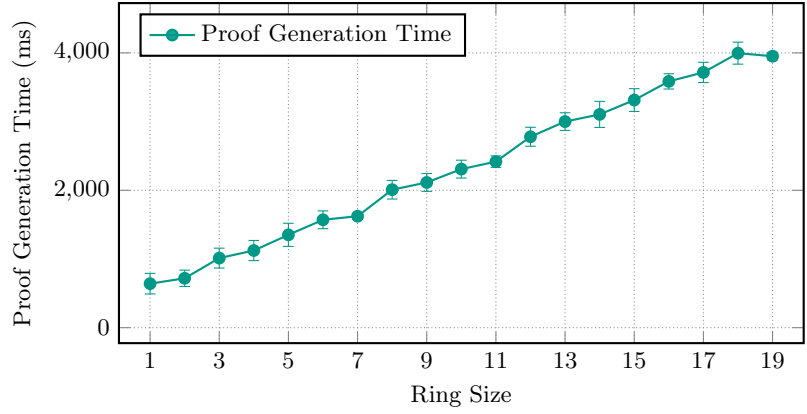

\begin{figure}[htbp]
\centering
\begin{tikzpicture}
\begin{axis}[
    width=10.64cm,
    height=6.08cm,
    xlabel={Ring Size},
    ylabel={Proof Size (bytes)},
    xmin=1, xmax=19,
    ymin=0, ymax=2000,
    grid=major,
    legend pos=north west,
]

\addplot[
    color=sciMagenta, 
    thick, 
    mark=*, 
    mark options={solid, fill=sciMagenta}
]
coordinates {
(1,129.00) (2,225.00) (3,321.00) (4,417.00) (5,513.00)
(6,609.00) (7,705.00) (8,801.00) (9,897.00) (10,993.00)
(11,1089.00) (12,1185.00) (13,1281.00) (14,1377.00) (15,1473.00)
(16,1569.00) (17,1665.00) (18,1761.00) (19,1857.00)
};
\addlegendentry{Proof Size}

\end{axis}
\end{tikzpicture}
\caption{Proof size growth as a function of ring size. The packed \texttt{zk\_proof} payload includes the ring size, $n$ public commitments, an initial challenge, and $n$ scalar responses, resulting in a deterministic size of $96n + 33$ bytes. The strict 2048-byte \texttt{MaxAppTotalArgLen} limit of the AVM dictates the practical upper bound of $n=19$ for a single verification call.}
\label{fig:proof_size}
\end{figure}
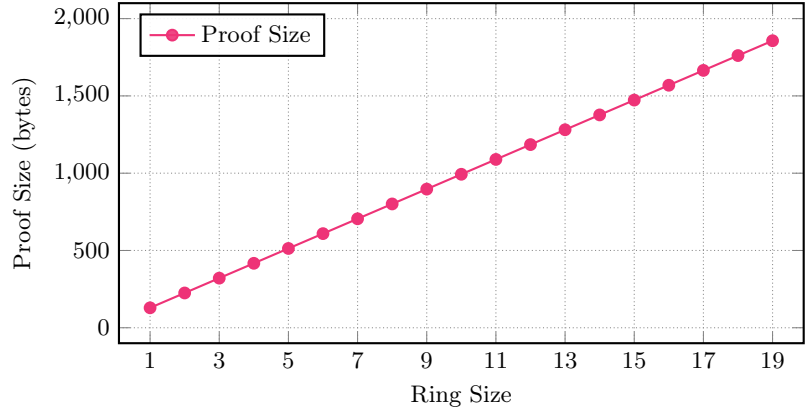

\begin{figure}[htbp]
\centering
\begin{tikzpicture}
\begin{axis}[
    width=10.64cm,
    height=6.08cm,
    xlabel={Ring Size},
    ylabel={Time (ms)},
    xmin=1, xmax=19,
    ymin=0, ymax=8000,
    grid=major,
    legend style={at={(0.5,-0.22)}, anchor=north, legend columns=2},
]

\addplot[
    color=sciBlue, thick, mark=*, mark options={solid, fill=sciBlue},
    error bars/.cd, y dir=both, y explicit
] coordinates {
(1,536.60) +- (0,147.66) (2,588.40) +- (0,123.49) (3,593.80) +- (0,114.64) 
(4,531.60) +- (0,143.66) (5,531.00) +- (0,145.49) (6,478.40) +- (0,129.34) 
(7,550.80) +- (0,125.01) (8,487.00) +- (0,135.87) (9,632.40) +- (0,12.84) 
(10,431.20) +- (0,114.83) (11,390.20) +- (0,5.81) (12,540.80) +- (0,147.77) 
(13,582.40) +- (0,108.08) (14,575.20) +- (0,118.37) (15,490.80) +- (0,139.44) 
(16,584.60) +- (0,113.70) (17,481.20) +- (0,134.51) (18,485.40) +- (0,135.69) 
(19,541.00) +- (0,147.51)
};
\addlegendentry{Deposit Detail Generation}

\addplot[
    color=sciCyan, thick, densely dotted, mark=square*, mark options={solid, fill=sciCyan},
    error bars/.cd, y dir=both, y explicit
] coordinates {
(1,98.20) +- (0,11.82) (2,53.40) +- (0,28.77) (3,55.80) +- (0,29.73) 
(4,47.20) +- (0,18.59) (5,45.00) +- (0,12.43) (6,64.40) +- (0,28.55) 
(7,52.20) +- (0,25.97) (8,47.20) +- (0,19.93) (9,53.40) +- (0,29.42) 
(10,53.20) +- (0,27.79) (11,49.20) +- (0,21.84) (12,60.20) +- (0,47.29) 
(13,50.40) +- (0,17.36) (14,54.80) +- (0,30.36) (15,54.80) +- (0,30.78) 
(16,53.60) +- (0,27.68) (17,57.40) +- (0,39.87) (18,45.80) +- (0,17.20) 
(19,52.80) +- (0,29.48)
};
\addlegendentry{Transaction Construction}

\addplot[
    color=sciOrange, thick, dashed, mark=triangle*, mark options={solid, fill=sciOrange},
    error bars/.cd, y dir=both, y explicit
] coordinates {
(1,1262.20) +- (0,216.00) (2,1041.00) +- (0,111.03) (3,1090.60) +- (0,145.80) 
(4,1053.00) +- (0,157.77) (5,995.20) +- (0,64.18) (6,1019.80) +- (0,76.78) 
(7,990.00) +- (0,106.16) (8,999.60) +- (0,168.18) (9,1029.00) +- (0,111.98) 
(10,1003.80) +- (0,74.92) (11,1038.20) +- (0,117.98) (12,1010.00) +- (0,99.63) 
(13,999.40) +- (0,87.39) (14,999.20) +- (0,62.18) (15,975.20) +- (0,54.97) 
(16,1043.40) +- (0,117.67) (17,1035.20) +- (0,168.24) (18,1032.80) +- (0,53.26) 
(19,1010.60) +- (0,77.87)
};
\addlegendentry{Transaction Submission}

\addplot[
    color=sciPurple, thick, dashdotted, mark=diamond*, mark options={solid, fill=sciPurple},
    error bars/.cd, y dir=both, y explicit
] coordinates {
(1,4172.00) +- (0,1085.52) (2,4118.60) +- (0,701.48) (3,4764.60) +- (0,1018.30) 
(4,4407.60) +- (0,1046.39) (5,3779.80) +- (0,550.90) (6,4470.20) +- (0,819.64) 
(7,3516.00) +- (0,995.86) (8,4329.40) +- (0,718.94) (9,4077.20) +- (0,745.64) 
(10,4151.00) +- (0,894.57) (11,4454.40) +- (0,858.78) (12,4007.40) +- (0,843.20) 
(13,4496.00) +- (0,695.19) (14,4968.40) +- (0,605.83) (15,3840.80) +- (0,886.38) 
(16,3808.20) +- (0,936.58) (17,3489.80) +- (0,449.93) (18,3415.80) +- (0,370.25) 
(19,3975.60) +- (0,644.89)
};
\addlegendentry{Network Confirmation}

\addplot[
    color=sciRed, ultra thick, mark=x, mark options={solid},
    error bars/.cd, y dir=both, y explicit
] coordinates {
(1,6069.80) +- (0,1123.00) (2,5802.00) +- (0,827.92) (3,6505.00) +- (0,898.90) 
(4,6040.00) +- (0,1037.41) (5,5351.40) +- (0,421.00) (6,6033.20) +- (0,830.33) 
(7,5109.40) +- (0,909.47) (8,5863.20) +- (0,635.63) (9,5792.40) +- (0,839.60) 
(10,5639.40) +- (0,880.62) (11,5932.40) +- (0,924.09) (12,5618.60) +- (0,807.99) 
(13,6128.40) +- (0,675.44) (14,6598.00) +- (0,609.58) (15,5362.00) +- (0,915.08) 
(16,5490.60) +- (0,1024.81) (17,5064.20) +- (0,618.41) (18,4979.40) +- (0,367.90) 
(19,5580.40) +- (0,863.23)
};
\addlegendentry{Total Deposit Time}

\end{axis}
\end{tikzpicture}
\caption{Breakdown of deposit latency versus ring size. The total deposit time is decomposed into four discrete operational stages: deposit detail generation, transaction construction, transaction submission, and network confirmation. Because the deposit phase involves a constant-time commitment generation and a standard atomic transaction group, all latency components remain independent of the ring size parameter $n$.}
\label{fig:deposit_breakdown}
\end{figure}
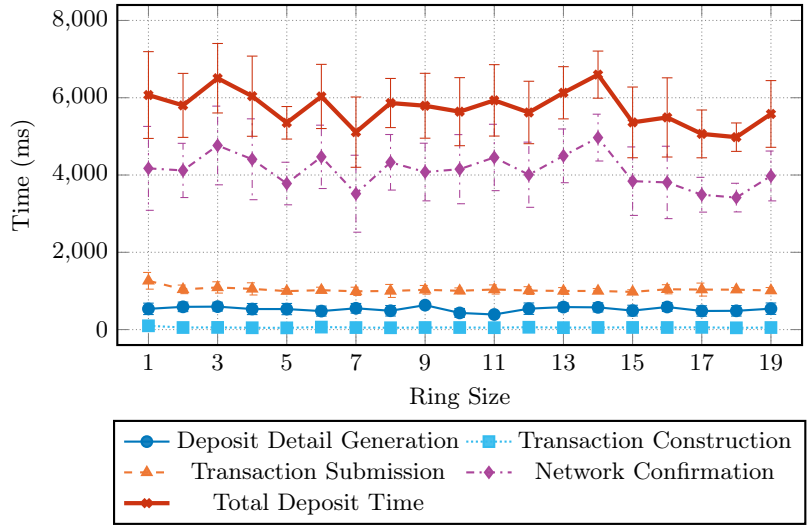

\begin{figure}[ht]
\centering
\begin{tikzpicture}
\begin{axis}[
    width=10.64cm,
    height=6.08cm,
    xlabel={Ring Size},
    ylabel={Time (ms)},
    xmin=1, xmax=19,
    ymin=0, ymax=12000,
    grid=major,
    legend style={at={(0.5,-0.22)}, anchor=north, legend columns=2},
]

\addplot[
    color=sciTeal, thick, mark=*, mark options={solid, fill=sciTeal},
    error bars/.cd, y dir=both, y explicit
] coordinates {
(1,602.20) +- (0,172.43) (2,667.60) +- (0,153.25) (3,563.80) +- (0,166.20) 
(4,806.60) +- (0,164.28) (5,798.20) +- (0,140.97) (6,810.60) +- (0,185.89) 
(7,949.60) +- (0,27.45) (8,818.60) +- (0,192.37) (9,924.00) +- (0,160.24) 
(10,940.60) +- (0,160.89) (11,1120.20) +- (0,43.10) (12,1143.60) +- (0,295.04) 
(13,1176.20) +- (0,656.07) (14,1150.40) +- (0,160.51) (15,1242.20) +- (0,90.65) 
(16,1338.40) +- (0,253.68) (17,1302.40) +- (0,151.28) (18,1241.60) +- (0,172.01) 
(19,1445.60) +- (0,45.11)
};
\addlegendentry{Withdrawal Data Generation}

\addplot[
    color=sciBlue, thick, densely dotted, mark=square*, mark options={solid, fill=sciBlue},
    error bars/.cd, y dir=both, y explicit
] coordinates {
(1,640.80) +- (0,150.09) (2,719.00) +- (0,118.64) (3,1013.00) +- (0,144.52) 
(4,1123.20) +- (0,146.84) (5,1351.40) +- (0,168.66) (6,1570.60) +- (0,129.62) 
(7,1622.60) +- (0,40.67) (8,2008.40) +- (0,135.61) (9,2114.00) +- (0,129.42) 
(10,2308.20) +- (0,129.86) (11,2417.20) +- (0,83.33) (12,2779.40) +- (0,138.08) 
(13,3000.20) +- (0,128.20) (14,3104.60) +- (0,189.70) (15,3313.60) +- (0,165.42) 
(16,3585.40) +- (0,111.39) (17,3715.80) +- (0,147.34) (18,3996.20) +- (0,160.06) 
(19,3951.60) +- (0,55.58)
};
\addlegendentry{Proof Generation}

\addplot[
    color=sciMagenta, thick, dashed, mark=triangle*, mark options={solid, fill=sciMagenta},
    error bars/.cd, y dir=both, y explicit
] coordinates {
(1,32.20) +- (0,2.86) (2,32.20) +- (0,3.70) (3,31.00) +- (0,2.92) 
(4,32.60) +- (0,0.89) (5,32.60) +- (0,3.78) (6,33.80) +- (0,1.30) 
(7,33.20) +- (0,2.17) (8,36.00) +- (0,3.87) (9,31.80) +- (0,2.28) 
(10,33.20) +- (0,1.64) (11,34.40) +- (0,1.52) (12,34.20) +- (0,2.68) 
(13,32.80) +- (0,1.92) (14,32.20) +- (0,2.59) (15,33.60) +- (0,4.45) 
(16,34.40) +- (0,1.34) (17,32.00) +- (0,2.12) (18,33.40) +- (0,2.51) 
(19,33.60) +- (0,1.67)
};
\addlegendentry{Transaction Construction}

\addplot[
    color=sciOrange, thick, dashdotted, mark=diamond*, mark options={solid, fill=sciOrange},
    error bars/.cd, y dir=both, y explicit
] coordinates {
(1,837.00) +- (0,27.39) (2,876.60) +- (0,62.42) (3,847.40) +- (0,58.56) 
(4,817.80) +- (0,57.36) (5,843.20) +- (0,61.03) (6,953.20) +- (0,38.07) 
(7,952.00) +- (0,17.36) (8,1027.20) +- (0,101.33) (9,1101.00) +- (0,220.35) 
(10,971.00) +- (0,46.15) (11,1064.60) +- (0,72.95) (12,1064.00) +- (0,129.31) 
(13,1025.40) +- (0,55.48) (14,1248.40) +- (0,126.62) (15,1246.80) +- (0,97.86) 
(16,1233.60) +- (0,65.00) (17,1210.60) +- (0,143.75) (18,1241.80) +- (0,56.05) 
(19,1227.00) +- (0,58.19)
};
\addlegendentry{Transaction Submission}

\addplot[
    color=sciPurple, thick, mark=pentagon*, mark options={solid, fill=sciPurple},
    error bars/.cd, y dir=both, y explicit
] coordinates {
(1,3852.40) +- (0,894.44) (2,4232.80) +- (0,940.40) (3,4784.00) +- (0,365.73) 
(4,4651.80) +- (0,436.26) (5,4408.00) +- (0,451.86) (6,4061.80) +- (0,598.17) 
(7,3791.40) +- (0,747.30) (8,3613.60) +- (0,840.22) (9,3598.80) +- (0,550.24) 
(10,3706.00) +- (0,1033.27) (11,4244.20) +- (0,852.56) (12,4106.40) +- (0,844.41) 
(13,4355.00) +- (0,1260.41) (14,4628.40) +- (0,446.83) (15,4349.40) +- (0,870.07) 
(16,3963.20) +- (0,446.31) (17,3406.00) +- (0,425.66) (18,3601.80) +- (0,345.67) 
(19,3763.40) +- (0,548.85)
};
\addlegendentry{Network Confirmation}

\addplot[
    color=sciGray, ultra thick, dashed, mark=x, mark options={solid},
    error bars/.cd, y dir=both, y explicit
] coordinates {
(1,5965.40) +- (0,975.37) (2,6528.40) +- (0,931.19) (3,7239.40) +- (0,297.55) 
(4,7432.20) +- (0,459.43) (5,7432.80) +- (0,418.36) (6,7430.20) +- (0,496.49) 
(7,7348.60) +- (0,774.43) (8,7504.00) +- (0,804.76) (9,7770.20) +- (0,669.37) 
(10,7958.80) +- (0,981.28) (11,8880.40) +- (0,950.76) (12,9127.80) +- (0,803.57) 
(13,9589.60) +- (0,1019.29) (14,10163.80) +- (0,410.43) (15,10186.20) +- (0,723.91) 
(16,10154.60) +- (0,606.76) (17,9666.60) +- (0,415.93) (18,10114.60) +- (0,442.12) 
(19,10421.40) +- (0,571.55)
};
\addlegendentry{Total Withdrawal Time}

\end{axis}
\end{tikzpicture}
\caption{Breakdown of withdrawal latency versus ring size. The end-to-end withdrawal time is dominated by proof generation, which scales linearly $O(n)$ due to the cryptographic complexity of the LSAG signature. Transaction construction and submission times also exhibit slight growth for larger rings ($n > 6$) as the client must construct and sign additional grouped auxiliary transactions to bypass the AVM's Box Reference limits.}
\label{fig:withdrawal_breakdown}
\end{figure}
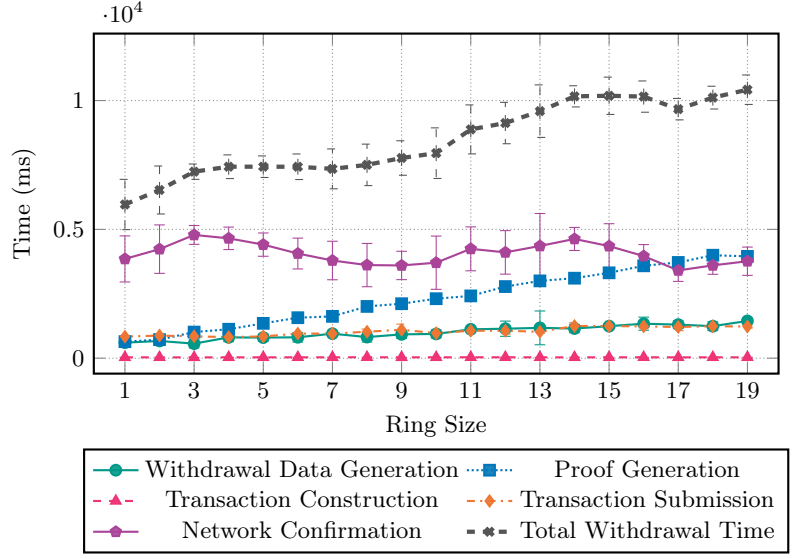

\paragraph{Scalability vs. Anonymity.}
The fundamental limitation of Obscura is the trade-off between execution feasibility and the size of the anonymity set. Because both proof size and verification cost scale as $O(n)$, the system cannot practically support the massive anonymity sets achieved by Merkle-based SNARK mixers. The current testnet implementation supports ring sizes up to $n=19$. While this provides significantly stronger combinatorial hiding than smaller rings, it remains vulnerable to sophisticated statistical deanonymization if the overall pool of users is small or highly correlated. Figure~\ref{fig:proof_size} and Figure~\ref{fig:proof_time} illustrate the linear scaling of proof payload size and generation time, respectively. Figure~\ref{fig:runtime_total} shows the end-to-end deposit and withdrawal latency. Figure~\ref{fig:deposit_breakdown} and Figure~\ref{fig:withdrawal_breakdown} provide a detailed breakdown of the latency for each phase of the protocol. To fully contextualize these measurements, the end-to-end execution is instrumented into discrete operational stages.

\paragraph{End-to-End Latency Breakdown.}
For the \textbf{deposit flow} (Figure~\ref{fig:deposit_breakdown}), the latency is decomposed into:
\begin{itemize}
    \item \emph{Deposit detail generation:} The duration required to generate the cryptographically secure random secret locally and invoke the backend service to compute the corresponding BN254 commitment.
    \item \emph{Transaction construction:} The time taken to fetch suggested transaction parameters from the Algorand node, construct the application call (carrying the commitment box reference) alongside the payment transaction, and group them atomically.
    \item \emph{Transaction submission:} The latency of signing the transaction group via the wallet provider and submitting the raw payload to the network.
    \item \emph{Network confirmation:} The network-level latency incurred waiting for the block to be proposed and the transaction to be confirmed on-chain.
\end{itemize}

For the \textbf{withdrawal flow} (Figure~\ref{fig:withdrawal_breakdown}), the latency incorporates the computationally intensive cryptographic generation:
\begin{itemize}
    \item \emph{Withdrawal data generation:} The time taken to retrieve contract state (reading commitment boxes), query the Indexer for recent deposits, select decoys, and construct the cryptographically shuffled anonymity set (the ring).
    \item \emph{Proof generation:} The exact duration the local backend enclave requires to compute the key image and generate the zero-knowledge LSAG ring signature proof using native BN254 curve operations. This phase exhibits the most significant $O(n)$ computational scaling.
    \item \emph{Transaction construction:} The time to decode the recipient address, fetch parameters, and build the withdrawal payload. For larger ring sizes, this includes the overhead of splitting box references across multiple grouped auxiliary transactions to bypass the AVM's consensus limits.
    \item \emph{Transaction submission:} The time to sign the withdrawal transaction (or transaction group) and submit it to the network.
    \item \emph{Network confirmation:} The latency from transaction broadcast to final on-chain confirmation.
\end{itemize}

\paragraph{Fixed Denominations and Amount Constraints.}
In the current testnet implementation, the transaction value is strictly fixed to 1 ALGO to ensure that all commitments within the anonymity set remain indistinguishable. While the system architecture could be extended in the future to support multiple parallel pools with different fixed denominations (e.g., separate 10 ALGO or 100 ALGO pools), arbitrary value selection within a single pool is not supported. Enabling flexible, hidden transaction amounts would require transitioning to a full Confidential Transaction (CT) scheme \cite{noether2016ring}. 
Such an extension would significantly increase the cryptographic overhead and verification complexity on the AVM. Consequently, users wishing to mix large volumes of funds must execute multiple discrete deposits and withdrawals, increasing their operational overhead and fee burden.

\paragraph{Payout Reductions.}
The withdrawal payout is not a 1:1 return of the deposited 1 ALGO. The smart contract deducts the network fees required for the inner \texttt{opup} burst, the inner payment transaction, and the Box Storage MBR. In the current reference implementation, the recipient receives $1{,}000{,}000 - 1{,}000 - 100{,}000 = 899{,}000$ micro-ALGO. Users must account for this $\sim$10\% protocol overhead when utilizing the protocol.

\subsection{Summary of Results Against the Criteria}
Table~\ref{tab:eval_criteria} defines the criteria; the measurements against them are summarized below.
\begin{itemize}
    \item \textbf{Proof size.} The measured payload grows from 129 bytes at $n=1$ to 1857 bytes at $n=19$, matching $96n+33$ exactly (Figure~\ref{fig:proof_size}). The argument-size limit, rather than the opcode budget, is therefore the binding constraint on $n$.
    \item \textbf{Execution cost.} The withdrawal issues $12n+1$ inner transactions, namely $12n$ \texttt{opup} calls and the inner payout, ranging from 13 to 229 (Figure~\ref{fig:inner_txs}). Withdrawals completed and were confirmed at every ring size from 1 to 19 (Figure~\ref{fig:withdrawal_breakdown}), which is consistent with the analytical result that $M=12$ suffices across the admissible range.
    \item \textbf{Transaction fees.} The deposit fee is constant at 0.001~ALGO, whereas the withdrawal fee grows linearly from 0.014~ALGO at $n=1$ to 0.232~ALGO at $n=19$ (Figure~\ref{fig:fees}), reflecting the inner-transaction count.
    \item \textbf{Proof generation.} Local signing time grows approximately linearly from about 0.64~s at $n=1$ to about 4.0~s at $n=19$ (Figure~\ref{fig:proof_time}).
    \item \textbf{Latency.} End-to-end deposit latency remains between approximately 5.0 and 6.6~s with no dependence on $n$. End-to-end withdrawal latency increases from about 6.0~s at $n=1$ to about 10.4~s at $n=19$ (Figure~\ref{fig:runtime_total}); roughly 3.3~s of this 4.5~s increase is attributable to proof generation. Network confirmation contributes roughly 3.4--5.0~s in both phases independently of $n$ (Figures~\ref{fig:deposit_breakdown} and~\ref{fig:withdrawal_breakdown}).
    \item \textbf{State and payout.} Each deposit and each withdrawal allocates one 64-byte box, so storage grows linearly with protocol usage rather than through accumulator updates, and the recipient of a 1~ALGO note receives 899,000 micro-ALGO.
\end{itemize}

\section{Conclusion and Future Work}
\label{sec:concl}
We introduced Obscura, a decentralized privacy protocol that establishes transaction anonymity on highly constrained smart contract platforms. By instantiating Linkable Spontaneous Anonymous Group (LSAG) signatures over the BN254 curve, we demonstrated that signer ambiguity and double-spend resistance can be enforced entirely on-chain without relying on trusted setups or native pairing operations. Our approach circumvents traditional virtual machine bottlenecks by replacing global Merkle accumulators with an $O(1)$ Box Storage state model, and it overcomes strict execution limits via dynamic opcode pooling. Obscura bridges the gap between high-throughput ledger architectures and the cryptographic requirements of non-custodial privacy pools, demonstrating that default transaction anonymity and user-controlled regulatory compliance can coexist.

While the protocol provides a functional and implementation-ready privacy layer, its reliance on $O(n)$ elliptic-curve verification inherently bounds the practical anonymity set size. We identify several critical directions for future research:

\begin{itemize}
    \item \textbf{Efficient Succinct Verification Without Pairings:} To scale anonymity sets, future work must explore zero-knowledge proof systems that offer sublinear verification. While the AVM provides support for native elliptic-curve pairings, their high computational cost in terms of opcode consumption, combined with state access contention arising from Merkle accumulator updates, makes the efficient deployment of traditional pairing-based zk-SNARK mixers challenging to scale efficiently. Adapting inner-product arguments (e.g., Bulletproofs \cite{bunz2018bulletproofs}) or hash-based STARKs \cite{ben2018scalable} for the specific opcode constraints of the AVM would significantly enhance combinatorial hiding without relying on trusted setups.
    
    \item \textbf{Variable Denominations:} The current protocol enforces a fixed denomination to maintain commitment indistinguishability. Extending the system to support arbitrary transaction amounts requires integrating Confidential Transactions \cite{maxwell2015confidential,noether2016ring}. Research is needed to efficiently implement Pedersen commitments \cite{pedersen1991non} and zero-knowledge range proofs \cite{bunz2018bulletproofs} within strict smart contract execution budgets.
    \item \textbf{Post-Quantum Resilience:} Because Obscura's soundness and privacy rely fundamentally on the hardness of the Elliptic Curve Discrete Logarithm Problem (ECDLP), it is vulnerable to future quantum adversaries. Transitioning the on-chain verification logic to support post-quantum ring signatures, such as lattice-based constructions \cite{esgin2019matrict,lu2019raptor}, remains a significant open challenge for smart contract privacy.
    \item \textbf{Formal Verification:} Finally, while the cryptographic primitives are well-studied, their instantiation in PyTeal and subsequent compilation to AVM bytecode may introduce potential implementation vulnerabilities. Formally verifying the correctness of the compiled bytecode and the underlying cryptographic algebra is essential to provide high-assurance security guarantees for the on-chain verifier.
\end{itemize}

\section*{Code Availability}
To facilitate reproducibility and public auditing, the complete source code supporting this work is available at \url{https://github.com/n-azimi/Obscura}. 

The repository provides the full end-to-end implementation of the Obscura protocol, including the PyTeal smart contract and associated verification script for on-chain execution, the Python-based LSAG cryptographic engine and local prover API, and the React client application responsible for decoy selection and transaction construction. Furthermore, the repository includes contract deployment scripts and analytical tooling (Obscura Inspector and Obscura Lens) for evaluating transaction topology and ledger state.

\appendix
\renewcommand{\theHsection}{appendix.\Alph{section}}
\section{Formal Security Model and Proofs}
\label{app:formal}

This appendix provides a formal cryptographic treatment of the Obscura protocol, adapting the security model of Linkable Spontaneous Anonymous Group (LSAG) signatures \cite{liu2004linkable} to our construction over the BN254 elliptic curve. We fix notation and computational assumptions, formalize the scheme syntax and its correctness, analyze the concrete challenge-hash instantiation, and reduce each claimed security property to a well-studied hard problem in the random oracle model (ROM) \cite{bellare1993random}.

\subsection{Notation}
Table~\ref{tab:notation} summarizes the notation used throughout this appendix. We note one convention inherited from the main text: the unsubscripted symbol $R$ always denotes the ring of public commitments, whereas the subscripted symbols $L_i, R_i \in \mathbb{G}$ denote the verification points of the $i$-th ring member.

\begin{table}[htbp]
\centering
\caption{Summary of Notation}
\label{tab:notation}
\begin{tabularx}{\textwidth}{@{}l X@{}}
\toprule
\textbf{Symbol} & \textbf{Description} \\
\midrule
$\lambda$ & Security parameter \\
$\mathbb{G}$, $q$ & Prime-order subgroup $\mathbb{G}_1$ of BN254 and its order, $q \approx 2^{254}$ \\
$G$, $H$ & Independent generators of $\mathbb{G}$ with $\log_G H$ unknown \\
$x$, $P = xG$ & Secret scalar and public commitment \\
$I = xH$ & Key image (nullifier) \\
$R$, $n$, $\pi$ & Ring of public commitments, ring size, secret signer index \\
$m$ & Signed message (32-byte recipient public key) \\
$\sigma = (c_0, s_0, \dots, s_{n-1})$ & LSAG signature \\
$\mathcal{H}$ & Challenge hash $\{0,1\}^* \to \mathcal{D}$, modeled as a random oracle \\
$\mathcal{D}$ & Challenge domain $[0, 2^{255})$ induced by the masked digest \\
$\mathcal{B}_C$, $\mathcal{B}_N$ & Commitment and nullifier boxes in on-chain storage \\
$q_H$, $q_S$ & Number of adversarial random-oracle and signing-oracle queries \\
\bottomrule
\end{tabularx}
\end{table}

\subsection{Computational Assumptions}

\begin{definition}[Elliptic Curve Discrete Logarithm Problem (ECDLP)]
Given a generator $G \in \mathbb{G}$ and a point $P \in \mathbb{G}$, it is computationally infeasible for any probabilistic polynomial-time (PPT) adversary $\mathcal{A}$ to output $x \in \mathbb{Z}_q$ such that $P = xG$; the success probability of any PPT $\mathcal{A}$ is negligible in $\lambda$.
\end{definition}

\begin{definition}[Decisional Diffie--Hellman (DDH) Assumption \cite{boneh1998decision}]
Given a tuple $(G, aG, bG, Z) \in \mathbb{G}^4$ for uniformly sampled $a, b \xleftarrow{\$} \mathbb{Z}_q$, it is computationally infeasible for any PPT adversary to distinguish whether $Z = abG$ or $Z = cG$ for a uniformly sampled $c \xleftarrow{\$} \mathbb{Z}_q$.
\end{definition}

Because BN254 is a pairing-friendly curve, the hardness of DDH in $\mathbb{G} = \mathbb{G}_1$ is precisely the \emph{external Diffie--Hellman} (XDH) assumption \cite{ballard2005correlation,boneh2004short}: DDH is conjectured hard in $\mathbb{G}_1$ of a Type-3 pairing configuration because no efficiently computable homomorphism from $\mathbb{G}_1$ to $\mathbb{G}_2$ is known. The self-pairing test that renders DDH easy in symmetric pairing groups is unavailable here, since the pairing requires one argument from each of $\mathbb{G}_1$ and $\mathbb{G}_2$, while all Obscura protocol elements lie in $\mathbb{G}_1$.

\begin{assumption}[Generator Independence]
\label{asm:generators}
The secondary generator $H$ is derived by applying a hash-to-curve map \cite{faz2023rfc} to a fixed, public domain-separation string, and this map is modeled as a random oracle into $\mathbb{G}$. Consequently, (i) no party knows $\log_G H$, and (ii) security reductions in the ROM may program the map, i.e., substitute a challenge point for $H$.
\end{assumption}

Assumption~\ref{asm:generators} is necessary, not merely convenient: if some party knew $k = \log_G H$, it could compute $I = kP$ for every on-chain commitment $P$ without knowledge of any secret scalar, linking all deposits to their nullifiers and collapsing the anonymity of the entire system. The nothing-up-my-sleeve derivation renders the assumption publicly verifiable.

\paragraph{Security Parameter Selection.}
The BN254 group order satisfies $q \approx 2^{254}$, so generic discrete-logarithm algorithms (e.g., Pollard's rho \cite{pollard1978monte}) require approximately $\sqrt{q} \approx 2^{127}$ group operations, providing a 127-bit security level for the ECDLP and XDH instances used above. We note that tower-number-field-sieve advances \cite{kim2016extended,barbulescu2019updating} degrade the security of the \emph{pairing} groups of BN254 (i.e., the DLP in the embedding field $\mathbb{F}_{p^{12}}$) to roughly 100 bits; Obscura, however, performs no pairing computations and relies solely on the elliptic-curve group law in $\mathbb{G}_1$, whose 127-bit security level is unaffected.

\subsection{Scheme Syntax and Correctness}

\begin{definition}[Linkable Ring Signature Scheme]
An LSAG signature scheme is a triple of PPT algorithms $(\textsf{KeyGen}, \textsf{Sign}, \textsf{Verify})$:
\begin{itemize}
    \item $(x, P) \leftarrow \textsf{KeyGen}(1^\lambda)$: outputs a secret scalar $x \xleftarrow{\$} \mathbb{Z}_q$ and the public commitment $P = xG$;
    \item $(\sigma, I) \leftarrow \textsf{Sign}(x, R, m)$: on input a secret key $x$, a ring $R$ containing $P = xG$ at a secret index $\pi$, and a message $m$, outputs a signature $\sigma = (c_0, s_0, \dots, s_{n-1})$ and the key image $I = xH$;
    \item $\{0,1\} \leftarrow \textsf{Verify}(R, m, I, \sigma)$: outputs $1$ (accept) or $0$ (reject).
\end{itemize}
The scheme is \emph{correct} if, for every message $m$, every honestly generated key pair $(x, P)$, and every ring $R$ containing $P$, $\textsf{Verify}(R, m, I, \sigma) = 1$ whenever $(\sigma, I) \leftarrow \textsf{Sign}(x, R, m)$.
\end{definition}

\begin{theorem}[Correctness]
\label{thm:correctness}
The Obscura instantiation of \textsf{Sign} and \textsf{Verify} (Section~\ref{sec:crypto}) is correct: an honestly generated signature is always accepted.
\end{theorem}

\begin{proof}
Let $\pi$ be the true signer's index. During signing, the prover sets $L_\pi = \alpha G$, $R_\pi = \alpha H$, and $c_{\pi+1} = \mathcal{H}(m \| L_\pi \| R_\pi)$; for each decoy index $i \neq \pi$ it samples $s_i$ and computes $L_i = s_i G + c_i P_i$, $R_i = s_i H + c_i I$, and $c_{i+1} = \mathcal{H}(m \| L_i \| R_i)$; finally it closes the ring with $s_\pi = \alpha - c_\pi x \pmod q$.

During verification, at every decoy index the verifier recomputes exactly the prover's points, since it uses the same $s_i$, $c_i$, $P_i$, and $I$. At the true index $\pi$ it computes
\begin{align*}
L_\pi &= s_\pi G + c_\pi P_\pi = (\alpha - c_\pi x)G + c_\pi (xG) = \alpha G, \\
R_\pi &= s_\pi H + c_\pi I = (\alpha - c_\pi x)H + c_\pi (xH) = \alpha H,
\end{align*}
which coincide with the prover's initialization. Hence every challenge in the verifier's chain equals the corresponding challenge in the prover's chain, and after the final iteration the running challenge equals $c_0$, so \textsf{Verify} accepts. \qed
\end{proof}

\subsection{Distribution of the Fiat--Shamir Challenges}
\label{app:challenge}
The protocol instantiates $\mathcal{H}(z) = \mathrm{SHA\text{-}256}(z) \wedge \mathrm{mask}$, where the mask clears the most significant bit of the digest. In the ROM, the output is therefore uniform over the domain $\mathcal{D} = [0, 2^{255})$. Since the BN254 group order satisfies $q < 2^{255} < 3q$, challenges are \emph{not} uniform modulo $q$; scalars $c \geq q$ are reduced implicitly by the \texttt{EcScalarMul} opcode. Correctness is unaffected, because the prover and the verifier compute identical masked digests and the final ring-closure check compares the 32-byte strings directly. The following lemma establishes that this instantiation also preserves every property required by the security proofs.

\begin{lemma}[Challenge Distribution]
\label{lem:challenge}
Model $\mathcal{H}$ as a random oracle with outputs uniform on $\mathcal{D} = [0, 2^{255})$. Then: (i) any challenge not yet queried is unpredictable except with probability $2^{-255}$; (ii) for any fixed $c \in \mathcal{D}$ and a fresh uniform $c' \in \mathcal{D}$, $\Pr[c' \equiv c \pmod q \,\wedge\, c' \neq c] \leq 3 \cdot 2^{-255}$, so the difference $c' - c$ is invertible modulo $q$ with overwhelming probability; and (iii) the joint distribution of the challenges is determined by the random oracle alone and is independent of the signer index.
\end{lemma}

\begin{proof}
(i) follows directly from the uniformity of the oracle output on $\mathcal{D}$. For (ii), the number of elements of $\mathcal{D}$ congruent to $c$ modulo $q$ is at most $\lceil |\mathcal{D}|/q \rceil = 3$, so a fresh uniform $c'$ collides with the residue class of $c$ with probability at most $3/2^{255}$; whenever $c' \not\equiv c \pmod q$, the difference is a nonzero element of $\mathbb{Z}_q$ and hence invertible, as $q$ is prime. For (iii), every challenge is the image of a distinct oracle query, and the oracle's answers are sampled independently of the signer index. \qed
\end{proof}

Property (i) guarantees that a forger cannot anticipate the ring-closure condition; property (ii) guarantees that the extraction steps of Theorems~\ref{thm:unforgeability} and~\ref{thm:linkability}, which divide by challenge differences modulo $q$, succeed with overwhelming probability; property (iii) underpins the signer-ambiguity argument of Theorem~\ref{thm:anonymity}.

\subsection{Formal Security Model}

The security of the Obscura protocol relies on three properties: existential unforgeability, signer ambiguity, and linkability. In all experiments the adversary $\mathcal{A}$ is a PPT algorithm with access to the random oracle $\mathcal{H}$.

\begin{definition}[Existential Unforgeability against Adaptive Chosen-Message and Chosen-Public-Key Attacks]
\label{def:unforgeability}
Let $\mathcal{SO}(R, m)$ be a signing oracle that accepts any public key list $R$ and any message $m$, and produces a valid signature $\sigma$. An LSAG signature scheme is existentially unforgeable if, for any PPT adversary $\mathcal{A}$ with access to $\mathcal{SO}$ and the random oracle $\mathcal{H}$, the probability that $\mathcal{A}$ outputs a valid forgery $(m^*, R^*, \sigma^*)$ is negligible. The forgery must satisfy $\textsf{Verify}(R^*, m^*, I^*, \sigma^*) = 1$, and $(m^*, R^*)$ must not correspond to any query--response pair of $\mathcal{SO}$. The interaction between the challenger and the adversary for this security experiment is illustrated in Figure~\ref{fig:euf_cma_exp}.
\end{definition}

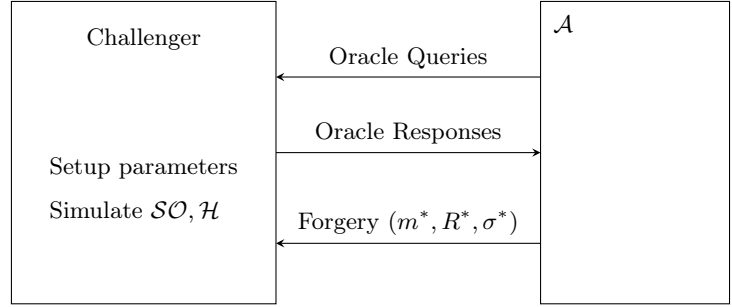
\begin{figure}[htbp]
    \centering
    \begin{tikzpicture}[>=stealth]
        \draw (0,0) rectangle (3.5,4);
        \node[align=center] at (1.75, 3.5) {Challenger};
        \node[align=left] at (1.75, 1.5) {
            Setup parameters \\[0.2cm]
            Simulate $\mathcal{SO}, \mathcal{H}$
        };

        \draw (7,0) rectangle (9.5,4);
        \node at (7.3, 3.7) {$\mathcal{A}$};

        \draw[<-] (3.5, 3.0) -- node[above] {Oracle Queries} (7, 3.0);
        \draw[->] (3.5, 2.0) -- node[above] {Oracle Responses} (7, 2.0);
        \draw[<-] (3.5, 0.8) -- node[above] {Forgery $(m^*, R^*, \sigma^*)$} (7, 0.8);
    \end{tikzpicture}
    \caption{Existential Unforgeability (EUF-CMA) Security Experiment.}
    \label{fig:euf_cma_exp}
\end{figure}

\begin{definition}[Signer Ambiguity]
\label{def:ambiguity}
Consider the following experiment, illustrated in Figure~\ref{fig:signer_ambiguity_exp}. The adversary $\mathcal{A}$ selects a message $m$, a ring $R$ of $n$ honestly generated public keys, and a set $D_t$ of $t$ corrupted secret keys ($0 \leq t < n-1$) whose corresponding public keys lie in $R$. The challenger samples the signer index $\pi$ uniformly from the non-corrupted indices, computes $(\sigma, I) \leftarrow \textsf{Sign}(x_\pi, R, m)$, and returns $(\sigma, I)$ to $\mathcal{A}$, who outputs a guess $\hat{\pi}$. An LSAG signature scheme is \emph{signer ambiguous} if for every PPT adversary $\mathcal{A}$,
\[
\Pr[\hat{\pi} = \pi] \leq \frac{1}{n-t} + \epsilon(\lambda)
\]
for some negligible function $\epsilon$.
\end{definition}

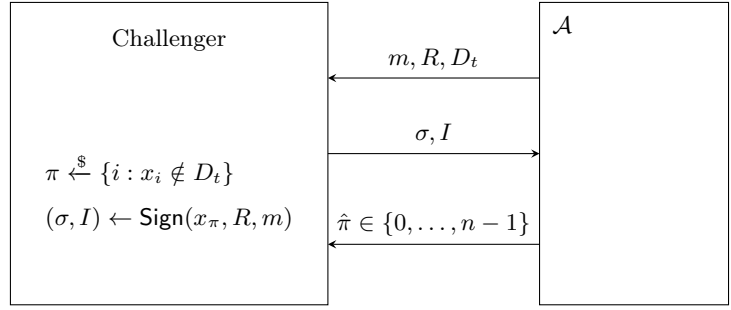
\begin{figure}[htbp]
    \centering
    \begin{tikzpicture}[>=stealth]
        \draw (0,0) rectangle (4.2,4);
        \node[align=center] at (2.1, 3.5) {Challenger};
        \node[align=left] at (2.1, 1.5) {
            $\pi \xleftarrow{\$} \{i : x_i \notin D_t\}$ \\[0.2cm]
            $(\sigma, I) \leftarrow \textsf{Sign}(x_\pi, R, m)$
        };

        \draw (7,0) rectangle (9.5,4);
        \node at (7.3, 3.7) {$\mathcal{A}$};

        \draw[<-] (4.2, 3.0) -- node[above] {$m, R, D_t$} (7, 3.0);
        \draw[->] (4.2, 2.0) -- node[above] {$\sigma, I$} (7, 2.0);
        \draw[<-] (4.2, 0.8) -- node[above] {$\hat{\pi} \in \{0, \dots, n-1\}$} (7, 0.8);
    \end{tikzpicture}
    \caption{Signer Ambiguity Security Experiment.}
    \label{fig:signer_ambiguity_exp}
\end{figure}

\begin{definition}[Linkability]
\label{def:linkability}
An LSAG signature scheme is \emph{linkable} if there exists a PPT algorithm \textsf{Link} that, given two valid signature tuples $(R, m_1, I_1, \sigma_1)$ and $(R, m_2, I_2, \sigma_2)$, outputs $1$ if and only if the two signatures were generated using the same secret key, except with negligible probability. In Obscura, $\textsf{Link}$ outputs $1$ if and only if $I_1 = I_2$. The property requires that no PPT adversary in possession of at most one secret key corresponding to the ring $R$ can output two valid signatures with distinct key images $I_1 \neq I_2$, nor a valid signature whose key image differs from $xH$ for the secret key $x$ it used.
\end{definition}

\subsection{Proof of Existential Unforgeability}

\begin{theorem}
\label{thm:unforgeability}
The Obscura LSAG signature scheme is existentially unforgeable against adaptive chosen-message and adaptive chosen-public-key attacks in the random oracle model, provided the ECDLP is hard in $\mathbb{G}$.
\end{theorem}

\begin{proof}
We construct a PPT reduction $\mathcal{M}$ that uses a successful forging adversary $\mathcal{A}$ to solve the ECDLP. Let $P^*$ be the ECDLP challenge point; $\mathcal{M}$ aims to find $x^*$ such that $P^* = x^*G$. $\mathcal{M}$ generates $N$ key pairs honestly, replaces the public key at a uniformly random position with $P^*$, publishes the resulting key list (all keys are identically distributed, so the embedding position is information-theoretically hidden from $\mathcal{A}$), and answers oracle queries as follows.

\textbf{Simulation of $\mathcal{H}$:} $\mathcal{M}$ answers by lazy sampling: for any new query it returns a uniform value from $\mathcal{D}$ and records the pair to maintain consistency across duplicated queries.

\textbf{Simulation of $\mathcal{SO}$:} When $\mathcal{A}$ requests a signature on message $m$ for a ring $R = \{P_0, \dots, P_{n-1}\}$, $\mathcal{M}$ simulates the signature without knowing any secret key by back-patching the random oracle. $\mathcal{M}$ selects an index $\pi \xleftarrow{\$} \{0, \dots, n-1\}$ and samples $c_0, s_0, \dots, s_{n-1}$ uniformly. To keep key images consistent with a real signer, $\mathcal{M}$ maintains a table assigning to each public key $P_i$ a fixed scalar $r_i \xleftarrow{\$} \mathbb{Z}_q$ and reuses the key image $I = r_\pi H$ across all queries simulated for $P_\pi$. For each $i \neq \pi$, $\mathcal{M}$ computes $L_i = s_i G + c_i P_i$ and $R_i = s_i H + c_i I$, obtaining $c_{i+1} = \mathcal{H}(m \| L_i \| R_i)$ from the simulated oracle. For the index $\pi$, $\mathcal{M}$ computes $L_\pi = s_\pi G + c_\pi P_\pi$ and $R_\pi = s_\pi H + c_\pi I$, and \emph{programs} the oracle so that $\mathcal{H}(m \| L_\pi \| R_\pi) = c_{\pi+1}$ (interpreting the index modulo $n$). Programming fails only if this input was previously queried; since $s_\pi$ is uniform, $L_\pi$ is uniform in $\mathbb{G}$, so a collision with any of the at most $q_H + nq_S$ recorded queries occurs with probability at most $(q_H + nq_S)/q$, which is negligible. The resulting transcript is distributed identically to a real signature.

\textbf{Extraction via Rewinding:} Suppose $\mathcal{A}$ outputs a valid forgery $(m, R, \sigma)$ with $\sigma = (c_0, s_0, \dots, s_{n-1})$. Except with negligible probability, all $n$ oracle queries used in verification appear among the recorded queries (otherwise $\mathcal{A}$ must have guessed an unqueried oracle output, which by Lemma~\ref{lem:challenge}(i) succeeds with probability at most $2^{-255}$). Let the \emph{gap index} $\pi$ denote the ring position whose challenge corresponds to the earliest of these queries. By the Forking Lemma \cite{pointcheval2000security}, $\mathcal{M}$ rewinds $\mathcal{A}$ to that critical query and supplies fresh oracle responses; with non-negligible probability $\mathcal{A}$ outputs a second valid forgery $\sigma' = (c'_0, s'_0, \dots, s'_{n-1})$ for the same $m$ and $R$, agreeing with the first transcript up to the fork, with $c'_\pi \neq c_\pi$.

Both transcripts satisfy the verification equations at the gap index with the same points $L_\pi, R_\pi$ (they are determined before the fork):
\begin{align*}
L_\pi &= s_\pi G + c_\pi P_\pi = s'_\pi G + c'_\pi P_\pi, \\
R_\pi &= s_\pi H + c_\pi I = s'_\pi H + c'_\pi I.
\end{align*}
By Lemma~\ref{lem:challenge}(ii), $(c'_\pi - c_\pi)$ is invertible modulo $q$ with overwhelming probability. Setting $x_\pi = (s_\pi - s'_\pi)(c'_\pi - c_\pi)^{-1} \bmod q$, the first equation yields $P_\pi = x_\pi G$, and the second yields $I = x_\pi H$: the extracted witness simultaneously opens the ring member \emph{and} the key image (this second identity is the basis of Theorem~\ref{thm:linkability}).

Since the forgery involves no corrupted keys and the embedding position of $P^*$ is uniform and independent of $\mathcal{A}$'s view, $P_\pi = P^*$ holds with probability at least $1/N$, in which case $\mathcal{M}$ outputs $x^* = x_\pi$ and solves the ECDLP. As the ECDLP is hard in $\mathbb{G}$, the forgery probability of $\mathcal{A}$ must be negligible. \qed
\end{proof}

\begin{corollary}[Concrete Security]
\label{cor:concrete}
Suppose $\mathcal{A}$ makes at most $q_H$ random-oracle queries and $q_S$ signing-oracle queries, and outputs a valid forgery with probability at least $1/Q(\lambda)$ for some polynomial $Q$. Then there exists an algorithm that solves the ECDLP in $\mathbb{G}$ with probability at least $\left(\frac{1}{n(q_H + nq_S)\,Q(\lambda)}\right)^{2}$ and expected running time at most twice that of $\mathcal{A}$, following the rewind-on-success analysis of \cite{liu2004linkable}.
\end{corollary}

\subsection{Proof of Signer Ambiguity}

\begin{theorem}
\label{thm:anonymity}
Under Assumption~\ref{asm:generators}, the Obscura LSAG signature scheme is signer ambiguous in the random oracle model, provided the DDH (XDH) problem is hard in $\mathbb{G}$.
\end{theorem}

\begin{proof}
We prove the case $t = 0$; the cases $0 < t < n-1$ follow by the same argument with the corrupted keys handed to the adversary. Suppose a PPT adversary $\mathcal{A}$ wins the experiment of Definition~\ref{def:ambiguity} with probability at least $\frac{1}{n} + \epsilon$ for non-negligible $\epsilon$. We construct a PPT distinguisher $\mathcal{M}$ for the DDH problem.

$\mathcal{M}$ is given a DDH challenge tuple $(G, \alpha G, \beta G, Z) \in \mathbb{G}^4$, where $Z$ is either $\alpha \beta G$ or a uniformly random point $\gamma G$, and must decide which. $\mathcal{M}$ programs the hash-to-curve oracle of Assumption~\ref{asm:generators} so that the secondary generator becomes $H = \beta G$. This step is where the reduction crucially relies on $H$ being a programmable random-oracle output rather than an arbitrary fixed constant; since $\beta G$ is uniformly distributed in $\mathbb{G}$, the programmed oracle is perfectly indistinguishable from an unprogrammed one.

$\mathcal{M}$ selects a target index $\pi \xleftarrow{\$} \{0, \dots, n-1\}$, sets $P_\pi = \alpha G$, generates all other ring members honestly as $P_i = x_i G$ with known $x_i$, and sets the key image $I = Z$. $\mathcal{M}$ then simulates a signature $\sigma$ for the ring $R$ and key image $I$ using the back-patching technique of Theorem~\ref{thm:unforgeability}, which requires no secret key, and runs $\mathcal{A}$ on $(m, R, \sigma, I)$.

$\mathcal{A}$ outputs a guess $\hat{\pi}$, and we distinguish the two cases of the challenge:
\begin{itemize}
    \item If $Z = \alpha \beta G$, then $I = \alpha(\beta G) = \alpha H$ is the well-formed key image for $P_\pi = \alpha G$, and the simulated transcript is distributed identically to a real signature by the member at index $\pi$: all responses $s_i$ are uniform, and by Lemma~\ref{lem:challenge}(iii) the challenges are random-oracle outputs whose distribution is independent of the signer index. By assumption, $\Pr[\hat{\pi} = \pi] \geq \frac{1}{n} + \epsilon$.
    \item If $Z = \gamma G$ for uniform $\gamma$, then $I$ is a uniformly random point independent of every ring member. Every component of $\mathcal{A}$'s view (the ring, the uniform responses, the oracle-derived challenges, and the independent key image) is distributed independently of $\pi$, so $\Pr[\hat{\pi} = \pi] = \frac{1}{n}$ exactly.
\end{itemize}

$\mathcal{M}$ outputs $1$ if and only if $\hat{\pi} = \pi$. Its distinguishing advantage is therefore
\[
\Pr[\mathcal{M} = 1 \mid Z = \alpha\beta G] - \Pr[\mathcal{M} = 1 \mid Z = \gamma G] \geq \left(\tfrac{1}{n} + \epsilon\right) - \tfrac{1}{n} = \epsilon,
\]
contradicting the hardness of DDH in $\mathbb{G}$. Hence $\epsilon$ must be negligible. \qed
\end{proof}

\begin{remark}[Culpability and Selective Disclosure]
When $t = n-1$, or whenever the true signer's secret key is disclosed, any party can test $x_i H \stackrel{?}{=} I$ and identify the signer with certainty; this is why Definition~\ref{def:ambiguity} requires $t < n-1$. Such \emph{culpability} is inherent to linkable ring signatures and is precisely the mechanism enabling the user-controlled auditability of Section~\ref{sec:sec}: by revealing $x$, a user proves both the deposit commitment $P = xG$ and the withdrawal key image $I = xH$ to an auditor, without affecting the ambiguity of any other participant's transactions.
\end{remark}

\subsection{Proof of Linkability}

\begin{theorem}
\label{thm:linkability}
The Obscura LSAG signature scheme is linkable in the random oracle model, provided the ECDLP is hard in $\mathbb{G}$.
\end{theorem}

\begin{proof}
The proof proceeds in two steps.

\emph{Step 1: every valid signature carries a well-formed key image.} The rewinding argument of Theorem~\ref{thm:unforgeability} applies to any algorithm producing a valid signature. Forking at the gap index $\pi$ yields two accepting transcripts sharing the points $L_\pi, R_\pi$, and hence the pair of equations
\begin{align*}
(s_\pi - s'_\pi)\, G &= (c'_\pi - c_\pi)\, P_\pi, &
(s_\pi - s'_\pi)\, H &= (c'_\pi - c_\pi)\, I,
\end{align*}
where $(c'_\pi - c_\pi)$ is invertible modulo $q$ by Lemma~\ref{lem:challenge}(ii). Writing $x_\pi = (s_\pi - s'_\pi)(c'_\pi - c_\pi)^{-1} \bmod q$, the first equation gives $P_\pi = x_\pi G$ and the second gives $I = x_\pi H$ \emph{with the same witness} $x_\pi$. Hence, except with negligible probability, any signature accepted by \textsf{Verify} has key image $I = x_\pi H$, where $x_\pi$ is the unique discrete logarithm of some ring member $P_\pi$.

\emph{Step 2: identical keys force identical key images.} Since $\mathbb{G}$ has prime order, the discrete logarithm of each ring member with respect to $G$ is unique in $\mathbb{Z}_q$. Suppose an adversary in possession of a single secret key $x$ outputs two valid signatures $\sigma_1, \sigma_2$ for rings containing $P = xG$. By Step 1, their key images satisfy $I_1 = x_{\pi_1} H$ and $I_2 = x_{\pi_2} H$ for extracted witnesses $x_{\pi_1}, x_{\pi_2}$. If both witnesses equal $x$, then $I_1 = xH = I_2$, and \textsf{Link} outputs $1$. Producing $I_1 \neq I_2$ therefore requires the two extracted witnesses to differ, i.e., the adversary must know the discrete logarithm of a \emph{second} ring member, which contradicts the ECDLP assumption for an adversary holding only one secret key. \qed
\end{proof}

\begin{corollary}[Double-Spend Resistance]
\label{cor:doublespend}
In the Obscura protocol, each deposit can be withdrawn at most once. By Theorem~\ref{thm:linkability}, every valid withdrawal spending the commitment $P = xG$ must present the key image $I = xH$. The smart contract asserts the non-existence of the nullifier box $\mathcal{B}_N$ before verification and records $I$ in $\mathcal{B}_N$ atomically with the payout; any second withdrawal attempt for the same deposit necessarily presents the identical $I$ and deterministically reverts, regardless of the decoy set chosen for the second ring.
\end{corollary}

\subsection{Protocol-Level Security}

We now lift the signature-level guarantees to the on-chain protocol, considering malicious users, malicious relayers and block proposers, and passive blockchain observers.

\begin{lemma}[Non-Frameability]
\label{lem:nonframe}
No PPT adversary can produce a valid signature whose key image equals $I = xH$ for an honest user's secret key $x$, except with negligible probability.
\end{lemma}
\begin{proof}
By Step 1 of Theorem~\ref{thm:linkability}, any valid signature with key image $I$ yields, under rewinding, a witness $x_\pi$ satisfying both $I = x_\pi H$ and $P_\pi = x_\pi G$ for some ring member $P_\pi$. Since discrete logarithms in a prime-order group are unique, $I = xH$ forces $x_\pi = x$, so a framing adversary can be used to extract the honest user's secret key, contradicting the ECDLP assumption. (Note also that computing $I$ directly from the public $P$ without $x$ is an instance of the Computational Diffie--Hellman problem, given Assumption~\ref{asm:generators}.) In particular, an adversary can neither preemptively ``spend'' an honest user's deposit nor cause an honest withdrawal to be falsely rejected as a double-spend. \qed
\end{proof}

\begin{lemma}[Replay Resistance and Destination Binding]
\label{lem:replay}
A valid withdrawal transaction cannot be replayed, and its recipient address cannot be altered, except with negligible probability.
\end{lemma}
\begin{proof}
The recipient address $m$ is bound into every challenge of the Fiat--Shamir chain, $c_{i+1} = \mathcal{H}(m \| L_i \| R_i)$. Substituting a different recipient $m' \neq m$ changes every random-oracle input; by Lemma~\ref{lem:challenge}(i), the probability that the mutated chain still closes (i.e., the recomputed final challenge equals $c_0$) is at most $2^{-255}$ per attempt. An active network adversary (e.g., a malicious relayer or block proposer) who intercepts a pending withdrawal therefore cannot redirect the payout without producing a fresh forgery, contradicting Theorem~\ref{thm:unforgeability}. Replaying the unmodified transaction is prevented unconditionally by Corollary~\ref{cor:doublespend}, as the key image $I$ is recorded in $\mathcal{B}_N$ upon first execution. \qed
\end{proof}

\paragraph{Blockchain Observers and Transaction Integrity.}
A globally passive adversary observes the ring $R$, the key image $I$, the recipient $m$, and the signature $\sigma$ of every withdrawal. By Theorem~\ref{thm:anonymity}, it identifies the true signer with probability at most $1/n$ plus a negligible term; the same reduction, together with Assumption~\ref{asm:generators}, shows that it cannot link any deposit commitment $P$ to its nullifier $I$. The AVM's atomic group execution further guarantees transaction integrity: the nullifier is recorded if and only if the signature verifies, and funds are disbursed if and only if the nullifier is recorded, leaving no reachable partial state for any interleaving of adversarial transactions.

\paragraph{Trust Assumptions and Limitations.}
The guarantees above are conditional on the trust model of Section~\ref{sec:sec}: the \textsf{Sign} algorithm must execute in a trusted local environment, since any party observing the secret scalar $x$ trivially breaks both anonymity and non-frameability. The anonymity bound $1/n$ is a cryptographic ceiling; operational behavior (biased decoy selection, temporal correlation between deposits and withdrawals, small effective pools) can only degrade it, as discussed in Section~\ref{sec:sec}. Finally, all reductions in this appendix are classical; quantum adversaries, which invalidate the ECDLP and XDH assumptions, are outside the model (cf.\ the post-quantum discussion in Section~\ref{sec:sec}).

\paragraph{Complexity of Cryptographic Operations.}
The on-chain verification requires exactly $n$ iterations. Each iteration performs four scalar multiplications and two point additions on $\mathbb{G}$ (two of each for $L_i$ and for $R_i$), plus one hash evaluation, matching the per-member cost accounting of Table~\ref{tab:opcode_costs}. The computational complexity is therefore $O(n)$. The AVM's deterministic opcode pricing ensures that this complexity translates to a fixed execution cost of $7608n + O(1)$ opcodes. The dynamic opcode pooling mechanism via inner transactions scales linearly to meet this requirement, guaranteeing that the protocol remains executable up to the AVM's maximum argument size limit ($n=19$), without introducing unbounded loops or resource exhaustion vectors.

\bibliographystyle{splncs04}
\bibliography{ref}

\end{document}